\documentclass{article}

\usepackage{times}

\usepackage{colortbl} 
\usepackage{booktabs}
\usepackage{graphicx}
\usepackage{latexsym}\usepackage{amssymb}
\usepackage{amsmath}
\usepackage{amsthm}
\usepackage{authblk}
\usepackage{verbatim}
\usepackage{flushend, natbib} 

\usepackage{verbatim}

\usepackage{xspace}

\usepackage{enumerate}
\usepackage{color}

\newtheorem{theorem}{Theorem}
\newtheorem{proposition}[theorem]{Proposition}
\newtheorem{lemma}[theorem]{Lemma}
\newtheorem{corollary}[theorem]{Corollary}
\newtheorem{fact}[theorem]{Fact}
\newtheorem{example}{Example}
\newtheorem{definition}{Definition}

\renewcommand{\L}{\mathcal L}
\newcommand{\PS}{\text{\it PS}}

\newcommand{\B}{\text{\boldmath $B$}}
\newcommand{\N}{\mathcal N}
\newcommand{\T}{\mathcal T}
\newcommand{\Win}{\mathcal W}
\newcommand{\Veto}{\mathcal V}
\newcommand{\Winr}{\mathcal{W}^+}
\newcommand{\Vetor}{\mathcal{V}^+}

\newcommand{\I}{\mathcal I}

\newcommand{\D}{\mathcal D}
\newcommand{\J}{\mathcal A}

\newcommand{\tuple}[1]{\left\langle #1 \right\rangle}
\newcommand{\set}[1]{\left\{ #1 \right\}}
\newcommand{\prof}[1]{\text{\boldmath $#1$}}
\newcommand{\maj}{\mathsf{maj}}
\newcommand{\IFF}{\mbox{     \textsc{iff}     }}

\newcommand{\A}{\mathcal A}

\newcommand{\EA}{\mathcal{A}^\T}

\newcommand{\putaway}[1]{}

\renewcommand{\phi}{\varphi}

\begin{document}

\title{Negotiable Votes\thanks{This paper 
improves and extends work previously presented at the 24th International Joint Conference on Artificial Intelligence (IJCAI'15) \citep{GrandiEtAlIJCAI2015}. 
The paper has benefited from the feedback of the anonymous reviewers at IJCAI'15, COMSOC'14 and LOFT'14. The authors are greatly indebted to Edith Elkind and Ulle Endriss for 
valuable comments on earlier versions of this work, and to the participants of the Second International Workshop on Norms, Actions and Games (NAG'16) in Toulouse, the 2015 LABEX CIMI Pluridisciplinary Workshop on Game Theory in Toulouse, the 5th International Workshop on Computational Social Choice (COMSOC'14) in Pittsburgh, the 11th Conference on Logic and the Foundations of Game and Decision Theory (LOFT'14) in Bergen, the 12th Meeting of the Social Choice and Welfare Society held in Boston in 2014, the COST IC1205 Workshop on Iterative Voting and Voting Games held in Padova in 2014, the Workshop on Fair Division, Voting and Computational Complexity held in Graz in 2014, and the XXV Meeting of AILA held in Pisa in 2014, where this work has been presented.}}
\author[1]{Umberto Grandi}
\author[2]{Davide Grossi}
\author[3]{Paolo Turrini}

\affil[1]{\small{IRIT, University of Toulouse}} 
\affil[2]{\small{Department of Computer Science, University of Liverpool}}
\affil[3]{\small{Department of Computing, Imperial College London}}
\date{}
\maketitle


\maketitle

\begin{abstract}
We study voting games on binary issues, where voters hold an objective over the outcome of the collective decision and are allowed, before the vote takes place, to negotiate their voting strategy with the other participants. We analyse the voters' rational behaviour in the resulting two-phase game, showing under what conditions undesirable equilibria can be removed and desirable ones sustained as a consequence of the pre-vote phase.

\end{abstract}


\section{Introduction}

Group decision-making is a topic of increasing relevance for Artificial Intelligence (AI). Addressing the problem of how groups of self-interested, autonomous entities can take the `right' decisions together is key to achieve intelligent behaviour in systems dependent on the interaction of autonomous entities. Against this backdrop, social choice theory has by now become a standard tool in the toolbox of AI and, especially, multi-agent systems (henceforth, MAS), see, e.g.,  
\citep{shoham08multiagent,HandbookCOMSOC2015}.  Voting in particular has been extensively studied as a prominent group decision-making paradigm in MAS. Despite this, only a small and very recent body of literature \citep{DesmedtElkind2010,XiaConitzer2010,ObrazsovaEtAlSAGT2013,ElkindEtAlIJCAI2015} is starting to focus on voting as a form of strategic, non-cooperative, interaction. 


More specifically---and this is the focus of the current contribution---no work with the notable exception of the literature on iterative voting \citep{MeirEtAl2010}, \citep{XiaConitzer2010},\citep{ LevRosenscheinAAMAS2012},\citep{branzei13how},\citep{MeirEtAlEC2014},\citep{LevRosenscheinJAIR} has ever studied how voting behavior in rational agents is influenced by strategic forms of interaction that precede voting---like persuasion, or negotiation. Literature in social choice has recognised that interaction preceding voting can be an effective tool to induce opinion change and achieve compromise solutions \citep{dryzek03social,list11group} while in game theory pre-play negotiations are known to be effective in overcoming inefficient outcomes caused by players' individual rationality \citep{JW05}. 

%
%


\subsection{Rational, Truthful, but Inefficient Votes}

Consider a multiple referendum, where a group of voters cast yes/no opinions on a number of issues, 
which are then aggregated independently in order to obtain the group's opinion on those issues. An instance of this situation 
is represented in the following table:

\begin{center}
\begin{tabular}
{lccc}
& issue 1 & issue 2 & issue 3 \\
\toprule
voter 1 & $1$ & $0$ & $0$ \\
voter 2 & $0$ & $1$ & $0$ \\
voter 3 & $0$ & $0$ & $1$ \\
\midrule
\rowcolor[gray]{.9}
Majority & $0$ & $0$ & $0$
\end{tabular}
\end{center}
\smallskip
where three voters express their binary opinions on each of three issue (1 for acceptance, 0 for rejection), which are aggregated one by one using the majority rule. Voters typically approach such a referendum with some preference over its outcomes. So, to continue on the example above, let us assume that each voter $i$ is interested in having the group accept issue $i$. Given these goals, it is rational for each voter to cast a truthful vote, that is a vote in which each voter $i$ accepts issue $i$: if the voter is not pivotal (that is, it belongs to a minority), its vote will not count, but if it is pivotal, casting a truthful vote will make its own opinion become majority. The example---which is also an instance of the so-called Ostrogorski paradox \citep{anscombe76frustration,daudt76ostrogorski}---shows a situation in which truthful voting leads to an inefficient majority. That is, a majority that rejects all issues and in so doing fails to meet the goal of each voter. Importantly, in the example there are a number of truthful profiles that would lead to an efficient majority (e.g., the profile where each voter accepts every issue). 
This is a deep tension within (binary) voting: when sincere voting is rational (weakly dominant) it may turn out to be inefficient. As we will see, this is not a feature of the majority rule alone, but of a large class of well-behaved aggregation rules. Understanding mechanisms which can resolve such inefficiencies is, we argue, an important step towards the development of human-level group decision-making capabilities in artificial intelligence.

\subsection{Contribution \& Scientific Context}

In this paper we study pre-vote negotiations in voting games over binary issues, where voters hold a simple type of lexicographic preference over the set of issues: they hold an objective about a subset of them while they are willing to strike deals on the remaining ones. Voters can influence one another before casting their ballots by transferring utility in order to obtain a more favourable outcome in the end. We show that this type of pre-vote interaction has beneficial effects on voting games by refining their set of equilibria, and in particular by guaranteeing the efficiency of truthful ones. Specifically, we isolate precise conditions under which `bad' equilibria---i.e., truthful but inefficient ones---can be overcome, and `good' ones can be sustained. Our work relates directly to several on-going lines of research in social choice, game theory and their applications to MAS.

\paragraph{Binary Aggregation} 
Aggregation and merging of information is a long studied topic in AI \citep{konieczny02merging,chopra06social,everaere07strategy} and judgment aggregation has become an influential formalism in AI \citep{endriss16judgment}. This setting is also known as {\em voting in multiple referenda} \citep{LacyNiou2000}. Binary voting can be further enriched by imposing that individual opinions also need to satisfy a set of integrity constraints, like in binary voting with constraints \citep{GrandiEndrissAIJ2013} and judgment aggregation \citep{dietrich07arrow,GrossiPigozzi2014}. 
Standard preference aggregation, which is the classical framework for voting theory, is a special case of binary voting with constraints \citep{dietrich07arrow}. Voting with constraints will be touched upon towards the end of the paper.
Research in binary voting and judgment aggregation focused on the (non-)manipulability of judgment aggregation rules \citep{DietrichListEP2007,BotanEtAlAAMAS2016} and its computational complexity \citep{EndrissEtAlJAIR2012,BaumeisterEtAlMMS2015}, but a fully-fledged theory of non-cooperative games in this setting has not yet been developed and that is our focus here. 

\paragraph*{Election Control}
The field of computational social choice has extensively studied decision problems involved with various forms of election control (see \citep{faliszewski16control}, for a recent overview) such as: lobbying \citep{ChristianEtAl2007,BredereckJAIR2014}, bribery \citep{BaumeisterEtAlMMS2015,HazonEtAlIJCAI2013}, modelled from the single agent perspective of a lobbyist or briber who tries to influence voters' decisions through monetary incentives, or from the perspective of a coalition of colluders \citep{bachrach11coalitional}. Here we study a form of control akin to bribery, but where any voter can 'bribe' any other voter. Our work can be seen as an effort to develop a fully fledged game-theoretic model of this type of control. Our focus is on equilibrium analysis and we therefore sidestep issues of computational complexity in this paper.

\paragraph{Equilibrium refinement}
Non-cooperative models of voting are known to suffer from a multiplicity of equilibria, many of which appear counterintuitive, not least because of their inefficiency. Equilibrium selection or refinement is a vast and long-standing research program in game theory \citep{meyerson78refinements}. Models of equilibrium refinement have been applied to voting games in the literature on economics \citep{gueth91majority,kim96equilibrium} and within MAS especially within the above-mentioned iterative voting literature \citep{DesmedtElkind2010,ObrazsovaEtAlSAGT2013,ObraztsovaEtAlAAAI2015}, which offers a natural strategy for selecting equilibria through the process of best response dynamics that starts from a profile of truthful votes. Our model tackles the same issue of refinement of equilibria in the context of binary voting, and focusing on those equilibria that are truthful and efficient. Unlike in iterative voting, our model is a two-phase model where equilibria are selected by means of an initial pre-vote negotiation phase, followed by voting.

\paragraph{Boolean Games}
We model voting strategies in binary aggregation with a model that generalises the well-known boolean games \citep{harrenstein01boolean,wooldridge13incentive}: voters have control of a set of propositional variables, i.e., their ballot, and have specific goal outcomes they want to achieve. In our setting the goals of individuals are expressed  \emph{over the outcome} of the decision process, thus on variables that---in non degenerate forms of voting---do not depend on their single choice only. Unlike boolean games, where each actor uniquely controls a propositional variable, in our setting the control of a variable is shared among the voters and its final truth value is determined by a voting rule. A formal relation with boolean games will be provided towards the end of the paper. 

\paragraph{Pre-play Negotiations} 
We model pre-vote negotiations as a pre-play interaction phase, in the spirit of \citep{JW05}. During this phase, which precedes the play of a normal form game---the voting game--- players are entitled to sacrifice part of their final utility in order to convince their opponents to play certain strategies, which in our case consist of voting ballots. In doing so we build upon the framework of endogenous boolean games \citep{Turrini16}, which enriches boolean games with a pre-play phase. We abstract away from the sequential structure of the bargaining phase, modelled in closely related frameworks \citep{GorankoT16}. 

\subsection{Outline of the Paper}

The paper is organised as follows.
In Section~\ref{sec:preliminaries} we present the setting of binary aggregation, defining the (issue-wise) majority rule as well as general classes of aggregation procedures, which constitute the rules of choice for the current paper. 
In Section~\ref{sec:J0} we define voting games for binary aggregation, specifying individual preferences by means of both a goal and a utility function, and we show how undesirable equilibria can be removed by appropriate modifications of the game matrix.
In Section~\ref{sec:J2} we present a full-blown model of collective decisions as a two-phase game, with a negotiation phase preceding voting. 
We show how the set of equilibria can be refined by means of rational negotiations.
Section~\ref{section:extendedgoals} relaxes assumptions we make on voters' goals in the basic framework showing the robustness of our results.
Section \ref{sec:related} discusses related work in some more detail. Finally, Section~\ref{sec:conclusions} concludes.


\section{Preliminaries: Binary Aggregation}\label{sec:preliminaries}


The study of binary aggregation dates back to \citep{wilson75theory}, and was recently revived both within economics \citep{DokowHolzman2010} and AI
\citep{GrandiEndrissAIJ2013}. This section is a brief introduction to the key notions and definitions of the framework.

\subsection{Basic Definitions}

In binary aggregation a finite set of agents (to which we will also refer as voters and, later on, players) $\N=\{1,\dots,n\}$ express yes/no opinions on a finite set of binary issues $\I=\{1,\dots,m\}$, and these opinions are then aggregated into a collective decision over each issue. This is analogous to voting over a simple (binary) combinatorial domain \citep{lang16voting}. Issues can be seen as queries to voters in a multiple referendum \citep{LacyNiou2000}, or seats that need to be allocated to candidates---two per seat---in assembly composition \citep{benoit10only}, or candidates of which voters approve or disapprove of \citep{brams78approval}.

%
In this paper we assume that $|\N|\geq 3$, i.e., there are at least 3 individuals.
We denote with $\D= \{B \mid B: \I \to \{0,1\}\}$ the set of all possible binary opinions over the set of issues $\I$ and call an element $B\in\D$ a {\bf ballot}. $B(j)=0$ (respectively, $B(j)=1$) indicates that the agent who submits ballot $B$ rejects (respectively, accepts) the issue $j$. A {\bf profile} $\B=(B_1,\dots,B_n)$ is the choice of a ballot for every individual in $\N$. Given a profile $\B$, we use $B_i$ to denote the ballot of individual $i$ within a profile $\B$. We adopt the usual convention writing $-i$ for $\N \backslash \set{i}$ and thus $\B_{-i}$ to denote the sequence consisting of the ballots of individuals other than $i$. Thus, $B_{i}(j)=1$ indicates that individual $i$ accepts issue $j$ in profile $\B$. Furthermore, we denote by $\N^{\B}_j=\{i \in \N \mid B_{i}(j)=1\}$ the set of individuals accepting issue $j$ in profile $\B$.




\medskip

Given a set of individuals $\N$ and issues $\I$, an {\bf aggregation rule} for $\N$ and $\I$ is a function $F:\D^\N \to \D$, mapping every profile to a binary ballot in~$\D$, called the {\bf collective} ballot. 
$F(\B)(j) \in \set{0,1}$ denotes the value of issue $j$ in the aggregation of $\B$ via aggregator~$F$.

The  benchmark aggregation rule is the so-called {\bf issue-wise strict majority rule}, which we will refer to simply as (issue-wise) majority (in symbols, $\maj$). The rule accepts an issue if and only if a majority of voters accept it, formally $\maj(\B)(j)= 1$ if and only if $|\N^{\B}_j|\geq \frac{|\N|+1}{2}$. 
Other examples of aggregation rules include {\bf quota rules}, which accept an issue if the number of voters accepting it exceeds a given quota, possibly different for each issue. Formally, if a quota rule $F_q$ is defined via a function $q : \I \to \set{0,1,\ldots,n}$, associating a quota to each issue, by stipulating $F_q(\B)_j=1 \Leftrightarrow |N^{\B}_j| \geq q(j)$. $F_q$~is called {\bf uniform} in case $q$ is a constant function. Issue-wise majority is a uniform quota rule with $q = \left\lceil \frac{\N +1}{2} \right\rceil$.
A third class of rules are distance-based rules, which output the ballot that minimises the overall distance to the profile for a suitable notion of distance.\footnote{See \citep{GrossiPigozzi2014} and \citep{endriss16judgment} for a more detailed exposition of aggregation rules.}

Table~\ref{figure:discursive_dilemma} depicts a ballot profile with three agents (rows $A, B$ and $C$) on three issues (columns $p, q$ and $r$) and the aggregated ballot (fourth row) 
obtained by majority.

\begin{table}[t]
\begin{center}
\begin{tabular}{lccc}
& $p$ & $q$ & $r$ \\
\toprule
A & 1 & 0 & 1\\
B & 1 & 1 & 0\\
C & 0 & 0 & 0\\
\midrule
\rowcolor[gray]{.9}
$\maj$ & 1 & 0 & 0\\
\end{tabular}
\end{center}
\vspace{-0.3cm}
\caption{An instance of binary aggregation}
\label{figure:discursive_dilemma}
\end{table}


\subsection{Classes of Aggregators}

The vast majority of the results we present in this paper are formulated and proved with respect to classes of aggregators, rather than specific aggregators like majority. The classes we consider are standard in the literature that focuses on binary aggregation from an axiomatic standpoint. The following are the properties we will be working with.\footnote{We refer the reader to the relevant literature mentioned in the introduction for more information about the axiomatic method in binary and judgment aggregation.}
\begin{definition}\label{def:axioms}
An aggregator $F$ is said to be:
\begin{description}
\item[{\em independent}] if for any issue $j\in \I$ and any two profiles $\B,\B'\in\D$, if $B_{i}(j) =B'_{i}(j)$ for all $i\in\N$, then $F(\B)(j)=F(\B')(j)$;
\item[{\em neutral}] if for any two issues $j,k \in \I$ and any profile $\B\in \D$, if for all $i\in\N$ we have that $B_{i}(j) = B_{i}(k)$, then $F(\B)(j) = F(\B)(k)$;
\item[{\em responsive}] if for any issue $j \in \I$ there exist two profiles $\B$ and $\B'$ such that $F(\B)(j) = 1$ and $F(\B')(j) = 0$;
\item[{\em monotonic}] if for any issue $j \in \I$, $x \in \set{0,1}$ and any two profiles $\B, \B'\in\D$, if $B_{i}(j) = x$ entails $B'_{i}(j) = x$ for all $i\in\N$, and for some $\ell \in \N$ we have that $B_{\ell}(j) = 1 - x$ and $B'_{\ell}(j) = x$, then $F(\B)(j)= x$ entails $F(\B')(j)=x$.
\item[{\em anonymous}] if for any two players $i,j\in \N$ and any two profiles $\B, \B'\in\D$, if $B_{i} = B'_{j}$, $B'_{i} = B_{j}$ and, for each $k\in \N \setminus \{i,j\}$, we have that $B_k=B'_k$, then $F(\B)=F(\B')$.
\item[{\em dictatorial}] (a dictatorship) on issue $j$, if there exists $i \in \N$ such that for any profile $\B$, $F(\B)(j) = B_i(j)$.
A {\bf {\em dictatorship}} is an aggregator that is dictatorial on every issue.

%
%
\end{description}
Furthermore, an aggregator is said to be {\bf {\em systematic}} if it is both independent and neutral.
\end{definition}
An aggregator is independent if the decision of accepting a given issue $j$ does not take into consideration the judgment of the individuals on any issue other than $j$.\footnote{When the aggregator is independent, the process of aggregation is also referred to as proposition-wise voting in the literature on multiple referrenda \citep{sanver06ensuring}, or seat-by-seat voting in the literature on assembly composition \citep{benoit10only}, or as simultaneous voting in the literature on voting over combinatorial domains \citep{lang16voting}.}
It is responsive if it is not a constant function.
It is monotonic if increasing the support for a collective ballot on one issue does not modify the result.
It is neutral if all issues are treated in the same way, and anonymous if all voters are treated in the same way.
It is systematic if it is both independent and neutral.  Finally, an aggregator is dictatorial if whenever a designated individual accepts an issue, that issue is also collectively accepted and vice versa.

\medskip

The majority rule, like any uniform quota rule, satisfies systematicity, monotonicity and anonymity. Non-uniform quota rules satisfy independence, monotonicity and anonymity, but fail neutrality (and hence systematicity). The minority rule, i.e., the rule that always selects the opposite of the majority rule, satisfies independence, anonymity, neutrality, but fails monotonicity. Dictatorships are (trivially) systematic and monotonic, but obviously not anonymous.

In this paper we will especially focus on independent, monotonic and anonymous aggregators, swapping whenever possible independence with systematicity in order to obtain stronger results.

\subsection{Winning and Veto Coalitions}

Given an aggregator $F$, we call a set of voters $C \subseteq \N$ a {\bf winning coalition} for issue $j\in \I$ if for every profile $\B$ we have that {\em if} $B_i(j)=1$ for all $i\in C$ and $B_i(j)=0$ for all $i\not\in C$ {\em then} $F(\B)(j)=1$.\footnote{The definitions from this section are all standard, except for the notion of resilient coalitions treated below. See, for instance, \citep{DokowHolzman2010} \citep{DokowHolzman2010}.} The notion of winning coalition is closely related to the independence property defined above:\footnote{This and the following facts are assumed to be well-known, for a proof in the setting of judgment aggregation we refer to Endriss \citep{endriss16judgment}, Lemma 17.1.}

\begin{fact} \label{fact:ind}
An aggregator $F$ is independent iff for all $j\in\I$ there exists a set of subsets $\Win_{j}\subseteq \mathcal P(\N)$ such that, for each ballot $\B$, $F(\B)_j=1$ if and only $N_j^\B\in\Win_j$. I.e., $F$ is independent iff it can be defined in terms of a family $\set{\Win_j}_{j \in \I}$ of sets of winning coalitions.
\end{fact}
Furthermore, we call a set of voters $C \subseteq \N$ a {\bf  veto coalition} for issue $j\in \I$ if for every profile $\B$ we have that {\em if} $B_i(j)=0$ for all $i\in C$ and $B_i(j)=1$ for all $i\not\in C$ {\em then} $F(\B)(j)=0$. Clearly,
\begin{align}
C \in \Win_j & \IFF \N \backslash C \not\in \Veto_j. \label{eq:relation}
\end{align}
Observe that by the above relation and Fact \ref{fact:ind} an independent aggregator can equivalently be defined by a set $\set{\Veto_j}_{j \in \I}$ of veto coalitions, consisting of exactly the complements of those coalitions that do not belong to any $\Win_j$. If we want to make explicit the underlying aggregator we sometimes write $\Win_F$ (respectively, $\Veto_F$) for the set of winning (resp., veto) coalitions for $F$. 
Let us fix intuitions by a few examples. The sets of winning and veto coalition for issue-wise majority are, for every issue $j$: $\Win_j = \set{C \subseteq \N \mid |C| \geq \frac{|\N|+1}{2}}$ and $\Veto_j = \set{C \subseteq \N \mid |C| \geq \frac{|\N|}{2}}$. 
When $|\N|$ is odd, this is the only quota rule for which $\Win_j=\Veto_j$ (so-called unbiasedness).
For a constant aggregator which always accepts all issues, that is, a quota rule with $q = 0$, the sets of winning and veto coalitions are, for each issue $j$: $\Win_j = 2^\N$, i.e., any coalition is winning, and $\Veto_j = \emptyset$, i.e., no coalition is a veto coalition.

\smallskip

Additional properties imposed on an independent aggregator $F$ induce further structure on winning (and veto) coalitions:

\begin{fact}
An independent aggregator $F$ is neutral (and hence systematic) iff for each $j,k \in \I$ we have that $\Win_j=\Win_k$ (equivalently, $\Veto_j=\Veto_k$ for all $j,k\in \I$).
An independent aggregator $F$ is monotonic iff for each $j \in \I$ and for any $C \in \Win_{j}$ (respectively, $C \in \Veto_{j}$), if $C \subseteq C'$ then $C' \in \Win_j$ (respectively, $C' \in \Veto_{j}$), i.e., winning (and veto) coalitions are closed under supersets. 
An independent aggregator $F$ is anonymous iff for each $j \in \I$ we have that $C \in \Win_{j}$ implies that $D \in \Win_j$ whenever $|C|=|D|$, i.e., coalitions are winning only depending on their cardinality.
\end{fact}

We conclude this preliminary section with one last important definition.
We call $C$ a {\bf resilient winning coalition} (respectively, a {\bf resilient veto coalition}) for issue $j\in \I$ if $C$ is a winning (resp., veto) coalition for $j$ and, for every $i\in C$, $C\setminus\{i\}$ is also a winning (resp., veto) coalition for $j$.\footnote{This definition adapts the notion of resiliency of equilibria studied by \citep{halpern11beyond} to the notion of winning and veto coalition proper of binary aggregation. \label{footnote}}
For every issue $j$, the set of resilient winning (resp., veto) coalitions for $j$ is denoted $\Winr_j$ (resp., $\Vetor_j$).
Given a neutral aggregator $F$, we denote with $\Win$ (resp., $\Veto$) the set of winning (resp., veto) coalitions for $F$, and with $\Winr$ (resp., $\Vetor$) the set of resilient winning (resp., veto) coalitions for $F$.
Again if we want to make the underlying aggregator explicit, we write $\Winr_F$ (resp., $\Vetor_F$) for the set of resilient winning (resp., veto) coalitions for $F$.
For example, in the case of the majority rule we have $\Winr_{\maj} = \set{C \subseteq \N \mid |C| \geq \frac{|\N|+1}{2} + 1}$, i.e., all winning coalitions exceeding the majority threshold of at least one element are resilient winning coalitions. For a uniform quota rule $F_q$, the set of resilient winning coalitions is $\Winr_{F_q} = \set{C \subseteq \N \mid |C| \geq q + 1}$. Observe also that if $F$ is dictatorial, then $\Winr_F = \emptyset$.


\section{Games for Binary Aggregation}\label{sec:J0}


In this section we present a model of a strategic interaction played by voters involved in a collective decision-making problem on binary issues. Players' strategies consist of all binary ballots on the set of issues, and the outcome of the game is obtained by aggregating the individual ballots by means of a given aggregator. 
Players' preferences are expressed in the form of a simple goal that is interpreted on the outcomes of the aggregation (i.e., the collective decision), and by an explicit payoff function for each player $i$, yielding to $i$ a real number at each profile and encoding, intuitively, the material value $i$ would receive, should that profile of votes occur. We study the existence of equilibria of these games, paying particular attention to the truthful and efficient ones.


\subsection{Main Definitions}

Before defining aggregation games we need a last piece of notation. To each set of issues $\I$, we associate the set of propositional atoms $\PS = \set{p_1,\dots,p_{|I|}}$ containing one atom for each issue in $\I$. We denote with $\L_\PS$ the propositional language constructed by closing $\PS$ under a functionally complete set of boolean connectives (e.g., $\set{\neg, \wedge}$).

\subsubsection{Aggregation Games, Goals and Preferences}

\begin{definition}\label{def:JA1}
Let $\I$ and $\N$ be given. An {\em aggregation game} (for $\I$ and $\N$) is a tuple
$\mathcal A =  \tuple{\N,\I, F, \set{\gamma_i}_{i\in\N}, \set{\pi_i}_{i\in\N}}$ 
where:
\begin{itemize}
\item $F$ is an aggregator for $\N$ and $\I$;
\item each $\gamma_i$ is a cube, i.e. a conjunction of literals from $\L_\PS$,\footnote{Formally, 
each $\gamma_i$ is equivalent to $\bigwedge_{j\in K} \ell_j$ where $K\subseteq \I$ and $\ell_j=p_j$ or $\ell_j=\neg p_j$ for all $j\in K$.} which is called a {\em goal};
\item $\pi_i: \D^\N \to \mathbb R$ is a payoff function assigning to each strategy profile a real number representing the utility that player $i$ gets at that profile.
\end{itemize}
\end{definition}
Note that a strategy profile in an aggregation game is a profile of binary ballots, and will therefore be denoted with $\prof{B}$. 
In the context of aggregation games we will use the term ``strategy profile'' and ``ballot profile'', or even just ``profile'', interchangeably. 

\medskip

Goals, intuitively, represent properties of the outcome of the aggregation process that voters are not willing to compromise about. By making the assumptions that goals are cubes we assume that each voter has a simple incentive structure: she can identify a specific set of atoms that she wants to be true at the outcome, another set of atoms that she wants to be false, and she is indifferent about the value of the others. 
When comparing two outcomes, one of which satisfies her goal and one of which does not, a voter will choose the outcome satisfying her goal. 
Thus, the first degree of preference of agents is dichotomous \citep{ElkindLackner2015}. 
If then two outcomes both satisfy her goal, or both do not, then the voter will look at the value she obtains at those outcomes through her payoff function.

This, we argue, is a very natural class of preferences for binary aggregation. They are technically known as {\bf quasi-dichotomous} preferences \citep{wooldridge13incentive} and widely studied in the context of boolean games. Henceforth we employ the satisfaction relation $\models$  (respectively, its negation $\not\models$) to express that a ballot satisfies (respectively, does not satisfy) a goal. The preference relation induced on ballot profiles by goals and payoff function is defined as follows:

\begin{definition}[Quasi-dichotomous preferences] \label{def:pref}
Let $\A$  be an aggregation game. Ballot profile $\B$ is strictly preferred by $i \in \N$ over ballot profile $\B'$ (in symbols, ${\B \succ^{\pi}_i  \B^{\prime}}$) if and only if any of the two following conditions holds:
\begin{enumerate}[i)]
\item $F(\B^{\prime}) \not\models \gamma_i$ and $F(\B) \models \gamma_i$;
\item $F(\B^{\prime}) \models \gamma_i$ iff $F(\B) \models \gamma_i$, and $\pi_i(\B) > \pi^{\prime}_i(\B)$.
\end{enumerate}
The two ballots are equally preferred whenever $F(\B^{\prime}) \models \gamma_i$ iff $F(\B) \models \gamma_i$, and $\pi_i(\B) = \pi^{\prime}_i(\B)$. The resulting weak preference order among ballot profiles is denoted $\succeq^{\pi}_i$.
\end{definition}
In other words, a profile $\B$ is weakly preferred by player $i$ to $\B^{\prime}$ if either $F(\B)$ satisfies $i$'s goal and $F(\B^{\prime})$ does not or, if both satisfy $i$'s goal or neither do, but $\B$ yields to $i$ at least as good as payoff as $\B^{\prime}$.
Individual preferences over ballot profiles are therefore induced by their goals, by their payoff functions, and by the aggregation procedure used.

\medskip

Finally, goals relate in a clear way to the structure of winning and veto coalitions of an independent aggregator. One can identify for every goal which are the sets of agents that can force the acceptance of such goal:
\begin{definition} \label{def:goalwin}
For an independent aggregator $F$ and a cube $\gamma$ we say that $C$ is {\em winning for} $\gamma$ if and only if $C \in \Win_{j}$ for each $j$ such that $\gamma \models p_j$, and $C \in \Veto_{j}$ for each $j$ such that $\gamma \models \neg p_j$.\footnote{Symbol $\models$ is here used as logical entailment between two formulas.} We write $\Win_\gamma$ for the set of winning coalitions for $\gamma$. The set of resilient winning coalitions for $\gamma$ (denoted $\Winr_\gamma$) is defined in the obvious way.
\end{definition}
In words, a coalition is winning for a goal whenever it is winning for all the issues that need to be accepted for $\gamma$ to be satisfied, and veto for all the issues that need to be rejected for $\gamma$ to be true. An obvious adaptation of the definition yields the notion of veto coalition for a given goal which, however, we will not be using in this work.

\subsubsection{Classes of Aggregation Games}

A natural class of aggregation games is that of games where the individual utility only depends on the outcome of the collective decision: 
\begin{definition}
An aggregation game $\A$ is called {\bf uniform} if for all $i \in \N$ and profiles $\B$ it is the case that $\pi_i(\B)=\pi_i(\B')$ whenever $F(\B)=F(\B')$. 
A game is called {\bf constant}, if all $\pi_i$ are constant functions, i.e., for all $i\in\N$ and all profiles $\prof{B}$ we have that $\pi_i(\prof{B})=\pi_i(\prof{B}')$. 
\end{definition}
Clearly, all constant aggregation games are uniform. Games with uniform payoffs are arguably the most natural examples of aggregation games.
The payoff each player receives is only dependent on the outcome of the vote, and not on the ballot profile that determines it. For convenience, we assume that in uniform games the payoff function is defined directly on outcomes, i.e., $\pi_i:\D\to\mathbb R$. Constant games are games where players' preferences are fully defined by their goals, and are therefore dichotomous.

\smallskip

We call a strategy $B$ $i$-{\bf truthful} if it satisfies the individual's goal $\gamma_i$. 
Note that in the case in which $\gamma_i$ is a cube that specifies in full a single binary ballot---that is, the agent's goal {\em is} a ballot---our notion of truthfulness coincides to the classic notion used in judgment aggregation and binary voting where only one ballot is truthful, and all other ballots are available for strategic voting.

Let us introduce some further terminology on strategy profiles:

\begin{definition}\label{def:profiles}
Let $C \subseteq \N$. We call a strategy profile $\prof{B}=(B_1,\dots,B_n)$: 
\begin{enumerate}
\item $C$-truthful if all $B_i$ with $i \in C$ are $i$-truthful, i.e., $B_i\models\gamma_i$, for all $i \in C$;
\item $C$-goal-efficient ($C$-efficient) if $F(\prof{B})\models \bigwedge_{i \in C} \gamma_i$;\footnote{Observe that since each $\gamma_i$ is a cube, also $ \bigwedge_{i \in C} \gamma_i$ is a cube.}
\item totally $C$-goal-inefficient (totally $C$-inefficient) if $F(\prof{B}) \models \bigwedge_{i \in C} \neg \gamma_i$.
\end{enumerate}
\end{definition}

Observe that while the notion of truthfulness is a property of the ballot itself, with goals interpreted on the individual ballots to check for truthfulness, the two notions of efficiency are instead properties of the outcome of the aggregation. 
One of the main problems that will be encountered in the following sections is indeed the existence of truthful weakly dominant profiles that are however aggregated in totally inefficient outcomes.

\medskip

One last piece of notation: let us call a game $C$-{\bf consistent}, for $C \subseteq \N$, if the conjunction of the goals of agents in coalition $C$ is consistent, i.e., if the formula $\bigwedge_{i \in C} \gamma_i$ is satisfiable. 


\subsection{Equilibria in Aggregation Games}


In this section we study Nash equilibria (NE) in special classes of aggregation games (that is, those that are constant and uniform) and their properties. Our focus is the existence of `good' NE, that is, NE that are truthful and efficient. 



\subsubsection{Absence of Equilibria}

First of all, it is important to observe that aggregation games may not have, in general, (pure stragegy) NE.

\begin{fact} \label{prop:noexist}
There are aggregation games that have no NE.
\end{fact}

\begin{proof}
Define an aggregation game as follows.
Let $\I = \set{p}$, $\N = \set{1, 2, 3}$, and let $\gamma_1 = \gamma_2 = \gamma_3 =\top$. Assume also that the payoff function is defined as follows: for any ballot profile $\prof{B}$, we have that $\pi_1(\B) = 1$ if and only if $B_1 = B_2$, and 0 otherwise; $\pi_2(\B) = 1$ if and only if $B_1 \not= B_2$, and 0 otherwise; finally, $\pi_3$ is constant.
That is, agent $1$ wants $1$ and $2$ to agree on issue $p$ while agent $2$ wants them to disagree, and agent $3$ is indifferent among any two outcomes of the interaction.\footnote{Note that this game is not uniform.} It is easy to see that the aggregation game encodes a matching-pennies type of game between $1$ and $2$ and, therefore, the resulting  aggregation game does not have a NE.
\end{proof}
We will come back to the issue of the absence of NE in Section \ref{sec:J2}. 



\subsubsection{Equilibria in Constant Aggregation Games}

Recall that a strategy $B_i$ is {\bf weakly dominant} for agent $i$ if for all profiles $\prof{B}$ we have that $(\B_{-i},B_i) \succeq^{\pi}_i  \B$.
We begin with an important result providing sufficient conditions for truthful strategies in a constant aggregation game to be weakly dominant.

\begin{proposition}\label{lemma:weak}
Let $\A$ be a constant aggregation game with $F$ independent and monotonic, and let $i\in\N$ be a player. If a strategy $B_i$ is truthful then it is weakly dominant for $i$.\footnote{This proposition can also be obtained as corollary of a result by \citep{DietrichListEP2007}. We are indebted to Ulle Endriss for this observation.}
\end{proposition}

\begin{proof}
Let $B_i$ be a truthful strategy, i.e., $B_i \models \gamma_i$.  We want to show that $B_i$ is weakly dominant, that is for every $\prof{B}' \in \D^\N$,  $F(\prof{B}') \models  \gamma_i$ implies $F(\prof{B}'_{-i},B_i) \models \gamma_i$.
We proceed towards a contradiction and assume that, for some profile $\prof{B}^{\prime}$ we have that $F(\prof{B}') \models  \gamma_i$ and $F(\prof{B}'_{-i},B_i) \not\models \gamma_i$. 
Since by Definition \ref{def:JA1} individual goals are cubes, we have that $\gamma_i=\bigwedge_{j \in \I} \ell_j$, where $\ell_j$ is a literal built from $\PS$.
Hence there exists  a $k\in \I$ such that $F(\prof{B}'_{-i},B_i)\not\models \ell_k$ but $F(\prof{B}')\models \ell_k$. 
Assume w.l.o.g. that $\ell_k$ is positive, i.e., $\ell_k=p_k$. Since $B_i$ is assumed to be truthful, $B_i\models \ell_k$ (that is, $B_i(k)=1$). 
Now, $F$ is independent so the value of issue $k$ in the output of $F$ depends only on the values of $k$ in each individual ballot in the input profile. Moreover, since $F(\prof{B}')\models \ell_k$ and $B_i\models \ell_k$, by the monotonicity of $F$ we conclude that $F(\prof{B}_{-i},B_i)\models \ell_k$. Contradiction.
\end{proof}

Intuitively the proposition tells us that independent and monotonic aggregators, as far as only the satisfaction of individual goals is concerned, guarantee that players are always better off by casting a truthful ballot.
A first immediate consequence is that computing weak dominant strategies in constant aggregation games is polynomial, since it boils down to finding a satisfying assignment to the individual goal, which in our model is a conjunction of literals. Other consequences are stated in the following corollary:

\begin{corollary}
Let $\A$ be an aggregation game with $F$ independent and monotonic:
\begin{enumerate}[(i)]
\item any profile $\prof{B}$ such that $B_i\models \gamma_i$ for all $i\in \N$ is a NE;
\item if for all $i\in \N$ the formula $\gamma_i$ is consistent, then $\A$ has at least one NE.
\end{enumerate}
\end{corollary}

\medskip

The converse of Proposition~\ref{lemma:weak} can be obtained for stronger but still natural assumptions, providing sufficient conditions for weakly dominant strategies to be truthful.

\begin{proposition} \label{lemma:weakR}
Let $\A$ be a constant aggregation game with $F$ independent, monotonic, responsive and anonymous. Then all weakly dominant strategies for $i \in \N$ are truthful.
\end{proposition}

\begin{proof}
Consider a weakly dominant strategy $B_i$ for $i$, and assume towards a contradiction that $B_i\not \models \gamma_i$, i.e., $B_i$ is not truthful.
Since goals are cubes, there exists a $k$ such that $B_i\not \models \ell_k$, where $\gamma_i=\bigwedge_{j \in \I} \ell_j$. W.l.o.g., assume that $\ell_k = p_k$ and let us argue towards the acceptance of $k$. By responsiveness there exist $\B$ and $\B'$ such that $F(\B)(k)=1$ and $F(\B')(k)=0$. By anonymity, collective acceptance and rejection of $k$ depends only on the cardinality of the set of agents accepting $k$ in a given profile. So, by monotonicity there exists an integer $q$ ($\neq 0$, by responsiveness) such that $F$ accepts $k$ if and only if the cardinality of the set of agents accepting $k$ is at least $q$.\footnote{Readers acquainted with judgment and binary aggregation will have noticed we are working here with the class of quota rules with non-trivial quota. Cf. \citep{GrossiPigozzi2014}.} Therefore there exist $\B$ and $\B'$ such that $F(\B)(k)=1$ and $F(\B')(k)=0$ and such that $\B_{-i} = \B'_{-i}$ and $B_i(k) = 1 \neq B'_i(k)$. Furthermore, exploiting independence, let us assume that for any $j \neq i$ and $t \neq k$, $B_j(t) = B'_j(t)$ and $F(\B) \models \gamma_i$. That is, there exist two profiles such that the first satisfies $i$'s goal and the second does not, because the first meets the threshold for accepting $k$ while the second does not. And the reason for the second not meeting the threshold is $i$'s ballot to reject $k$. In other words, $i$ is pivotal in $\B'$ for the acceptance of $k$ and therefore for the acceptance of its goal.
It follows that for any strategy $B_i'$ such that $B_i'\models \ell_k$ we have that $F(\B_{-i},B_i')\models\gamma_i$, i.e., $B_i$ is dominated by $B'$, against the assumption that $B_i$ is weakly dominant in $\A$.
\end{proof}
Bundling Propositions \ref{lemma:weak} and \ref{lemma:weakR} together we thus obtain conditions over the aggregator (independence, monotonicity, responsiveness and anonymity) that imply the equivalence of truthfulness and weak dominance in constant aggregation games. As noted in the proof of Proposition \ref{lemma:weakR}, these conditions on the aggregator force it to be a quota rule (with non-$0$ quota). The equivalence of truthfulness and weak dominance is therefore a feature of an important class of well-behaved aggregators.

\medskip

Before concluding this section, it should be stressed that Proposition \ref{lemma:weak} ceases to hold if we allow the goals of the voters to be propositional formulas more complex than a cube:

\begin{example}\label{example:mono}
Let $F = \maj$, $\N = \I = \set{1,2,3}$ and let $\I = \set{1,2}$. Let then $\gamma_1 = p_1 \vee p_2$ and  $\gamma_2 = \gamma_3 = \top$. That is, agent $1$ is interested in having at least one of the two issues accepted, while the rest of the agents are indifferent.  We show that in this game not all all truthful ballots of $1$ are weakly dominant. Consider the profile $\B = (B_1,B_2,B_3)$ where $1$ votes the truthful ballot $B_1 = (0,1)$, $2$ votes $B_2 = (0,0)$ and $3$ votes $B_3 = (1,0)$. We have that $F(\B) \not\models \gamma_1$. Clearly, $1$ has a best response $B'_1 = (1,0) \neq B_1$ in that profile as $F(B'_1, B_2, B_3) \models \gamma_1$.
\end{example}

\begin{example}\label{example:cubes}
Let $F = \maj$, $\N = \I = \set{1,2,3}$ and let agent~$1$'s goal be that of having an odd number of accepted issues, while agents 2 and 3 have no specific goals.
Formally, let $\gamma_1=(p_1\wedge p_2\wedge p_3) \vee (p_1\wedge \neg p_2\wedge \neg p_3) \vee (\neg p_1\wedge p_2\wedge \neg p_3) \vee (\neg p_1\wedge \neg p_2\wedge p_3) $.  As above we show that not all all truthful ballots of $1$ are weakly dominant. Consider the profile $\B = (B_1, B_2, B_3)$ where $B_2=(1,0,0)$ and $B_3=(0,1,0)$ and where 
$1$ votes a truthful ballot $B_1=(0,0,1)$. This profile results under the majority rule in $(0,0,0)$. So the (non-truthful) ballot $B_1' = (1,0,1) \neq B_1$ is a better response in $\B$ for $1$ and yields the collective ballot $(1,0,0)$, which satisfies $\gamma_1$. 
\end{example}
In fact, cubes guarantee that players' preferences in constant aggregation games satisfy a property known as separability (cf. \citep{lang16voting} for separability in combinatorial domains), which has been shown in other voting frameworks to guarantee that truthful strategies are undominated \citep{benoit10only}.\footnote{Separability is normally defined over strict orders (but cf. \citep{hodge02separable} for a general treatment of the notion). In our setup separability can be defined as follows: $\preceq_i$ is separable if and only if for all $j \in \I$, if $\B \preceq_i \B'$ for two profiles such that $F(\B)(k) = F(\B')(k)$ for all $k \neq j$, then $\B'' \preceq_i \B'''$ for all profiles such that $F(\B'')(k) = F(\B''')(k)$ for all $k \neq j$ and $F(\B)(j) = F(\B'')(j)$. If $\preceq_i$ is a dichotomous preference induced by a cube (i.e., a conjunction of literals), then $\preceq_i$ is separable for the above definition.}  
  
\subsubsection{Truthfulness and Efficiency}  
  
Proposition~\ref{lemma:weak} establishes the existence of truthful equilibria. However, as we already noticed in the introduction of this paper, truthfulness does not guarantee efficiency. 
And this does not only hold for constant aggregation games based on majority, but for games based on a large class of aggregators---known to correspond to the class of all possible uniform quota rules\footnote{See \citep{DietrichList2007}, Theorem 1.}---as the following result shows.

\begin{proposition}\label{prop:badNE}
For every aggregator which is anonymous, systematic and monotonic, there exist constant aggregation games with truthful and totally inefficient NE in weakly dominant strategies.
\end{proposition}

\begin{proof}
The proof is by construction of a constant aggregation game $\A$ with the desired property. Let $F$ be anonymous, systematic and monotonic, let $\N = \I = \set{1,2,3}$. First we will construct two games and show that at least one of the two admits a truthful and totally inefficient profile. We can then apply Proposition \ref{lemma:weak} to show that such profile ought to be a NE in weakly dominant strategies. The games are built on the same aggregator $F$ and differ only on their goals. \fbox{Game A} Let, for each $i \in \N$, $\gamma_i = p_i$. Each $\gamma_i$ is therefore, trivially, a cube. Note also that $\bigwedge_{i \in \N} \gamma_i$ is satisfiable, so Game A is $\N$-consistent. \fbox{Game B} Let, for each $i \in \N$, the goal be defined as $\chi_i = \neg p_i$. Each $\chi_i$ is therefore, again, a trivial cube and Game B is also $\N$-consistent. Now construct a ballot profile $\B$ in Game A and a ballot profile $\B'$ in Game B as follows, for $i \in \N$ and $j \in \I$:
\begin{align}
B_i(j) = \left\{
\begin{array}{ll}
1 & \mbox{if $p_j = \gamma_i$} \\
0& \mbox{otherwise}
\end{array}
\right. & \ \ \ \ 
B'_i(j) = \left\{
\begin{array}{ll}
0 & \mbox{if $\neg p_j = \chi_i$} \\
1 & \mbox{otherwise}
\end{array}
\right.
\end{align}
That is, in $\B$ each voter votes $1$ only on the issue which coincides with its goal. Vice versa, in $\B'$ each voter votes $0$ only on the issues whose rejection coincides with its goal. By construction, $\B$ and $\B'$ are both truthful. 
Now assume that $\B$ is not totally inefficient, and we prove that $\B'$ is totally inefficient. So there exists $i \in \N$ such that $F(\B) \models \gamma_i$ and since by construction $\gamma_i = p_i$, $F(\B)(i) = 1$. By anonymity and systematicity, $\B$ must actually be efficient, that is $F(\B) \models \bigwedge_{i \in \N} \gamma_i$ as by construction the individual positions on each issue are identical, modulo permutations. More precisely, by construction for all $j \neq i \in \N$, $B_j(p_i) = 0$, that is, $i$ is the only voter accepting $p_i$ in $\B$. By independence and monotonicity it follows that $\set{i} \in \Win_{p_i}$, that is, if $i$ accepts $p_i$ so does $F$. 
However, $\set{i} \not\in \Veto_{p_i}$ for otherwise $i$ would be a dictator for $p_i$, against the assumption of anonymity for $F$. From $\set{i} \not\in \Veto_{p_i}$ and Equation \eqref{eq:relation} it follows that $\N \backslash \set{i} \in \Win_{p_i}$ from which we obtain that $F(\B')(p_i) = 1$ and hence that $F(\B') \models \neg \chi_i$. By the anonymity and systematicity of $F$ it therefore follows that $F(\B') \models \bigwedge_{i \in \N} \neg \chi_i$ as by construction the individual positions on each issue are identical, modulo permutations. $\B'$ is therefore a truthful but totally inefficient ballot profile. Since $\B'$ is, by construction, a profile where each voter is truthful, 
by Proposition \ref{lemma:weak}, $\B'$ is also a NE in weakly dominant strategies. This completes the proof.
\end{proof}


\subsubsection{Equilibria in Uniform Aggregation Games}


Since constant payoffs are special cases of uniform ones, negative results such as Proposition~\ref{prop:badNE} still hold for uniform aggregation games. 
As to truthful voting, Proposition~\ref{lemma:weak} does not generalise to uniform aggregation games:

\begin{proposition}\label{prop:nodominant}
There exist uniform aggregation games for $\maj$ in which truthful strategies are not dominant.
\end{proposition}

\begin{proof}
Define the uniform aggregation game as follows.
Let  $\I = \{p,q,t\}$, $\N = \{1,2,3\}$ and $F = \maj$. Let $\gamma_1=\neg p \wedge q \wedge \neg t$, $\gamma_2= \neg p \wedge \neg q \wedge \neg t$, and  $\gamma_3=\neg p \wedge \neg q \wedge t$.
Define the payoff functions as follows, let $\pi_i(B)=1$ for $i=3$ and $B=(0,1,0)$, and 0 otherwise. 
Take the following profiles: 
$\prof{B}_1=((0,1,0), (0,0,0), (0,0,1))$ and $\prof{B}_2=((0,1,0), (0,0,0), (0,1,0))$. 
Since $\maj(\prof{B}_1)=(0,0,0)$ and $\maj(\prof{B}_2)=(0,1,0)$, we have $\prof{B}_2 \succ^{\pi}_3 \prof{B}_1$ and  $\prof{B}_1$, unlike $\prof{B}_2$, contains a truthful strategy by $3$. 
\end{proof}
The fact that truthful voting is not always a dominant strategy for aggregation games with simple goals might seem counterintuitive, especially when the payoff is required to be uniform across profiles leading to the same outcome. The reason for this lies in the effect of the payoff function.
When a player is in the position of changing the outcome of the decision in a certain profile, this does not necessarily imply she has the power to make the collective decision satisfy her goal. She may only be able to lead the group to a decision which, even though still not satisfying her goal, yields a better payoff for her.

\medskip

Despite the negative result in Proposition~\ref{prop:nodominant}, we can still prove the existence of truthful and efficient equilibria in a uniform aggregation game if we assume the mutual consistency of the individual goals of a resilient winning coalition.

\begin{proposition}\label{prop:NE-existence}
Every $C$-consistent uniform aggregation game for $F$ independent and monotonic has a NE that is $C$-truthful and $C$-efficient, if $C\in \Winr_{\bigwedge\Gamma}$ where $\Gamma = \set{\gamma_i \mid i\in C}$.
\end{proposition}
\begin{proof}
Take a $C$-consistent uniform aggregation game. Notice that since each goal $\gamma_i$ is a cube, also $\bigwedge \Gamma$ is a cube, so that Definition \ref{def:goalwin} applies.
There exists then a ballot $B^*$ such that $B^*\models \bigwedge \Gamma$. 
Take now any ballot profile $\prof{B}^*$ such that $B^*$ is the ballot of all and only the voters in $C$ while all agents in $\N \backslash C$ vote the inverse ballot  $\overline{B^*}$ (that is, for any issue $j$, $B^*(j) = 1$ iff $\overline{B}^*(j) = 0$). Since $C\in \Win_{\bigwedge \Gamma}$ (by the assumption that $C\in \Winr_{\bigwedge \Gamma}$) we have that $F(\prof{B}^*)=B^*$. Clearly $F(\prof{B}^*)$ satisfies the conjunction of the goals of the individuals in $C$, and each individual in $C$ votes truthfully. We show that $\prof{B}^*$ is a NE, by showing that (a) no agent in $\N \backslash C$ has a profitable deviation, and (b) no agent in $C$ has a profitable deviation. As to (a), since $F$ is monotonic, any change in the ballot $\overline{B}^*(j)$ by some voter in $\N \backslash C$ does not change the outcome $F(\prof{B}^*)=B^*$. As to (b), any change in the ballot $B^*$ by some voter in $C$ does not change the outcome because $C\in \Winr_{\bigwedge \Gamma}$. 
\end{proof}


Let us recapitulate the findings of this section. We have shown, for a well-behaved class of aggregators, that aggregation games with constant payoffs have many NE, since the truthful ballots are exactly the weakly dominant strategies of the game (Propositions \ref{lemma:weak} and \ref{lemma:weakR}). This result does not carry over to uniform aggregation games, but we showed that with some stronger assumptions on the aggregator, the existence of truthful \emph{and} efficient equilibria in such games can be guaranteed each time a resilient winning coalition has non-conflicting goals (Proposition \ref{prop:NE-existence}). Proposition \ref{prop:badNE} has then highlighted a key issue of aggregation games: truthful equilibria may be totally inefficient in the sense of failing to satisfy the goals of all voters, even when all such goals are consistent. The purpose of the following section is to introduce an endogenous pre-play negotiation mechanism which, in equilibrium, allows players to select truthful and efficient NE of the underlying aggregation game, thus resolving the tension between truthfulness and efficiency.


\section{Pre-vote Negotiations}\label{sec:J2}


The section presents endogenous aggregation games, our model aggregation games augmented with a pre-vote negotiation phase. In a nutshell voters will now be allowed, before the vote takes place, to sacrifice a part of their expected gains in order to influence the other voters' decision-making. We show that allowing such negotiations: (i) guarantees the selection of efficient equilibria for all individuals when one such equilibrium exists (ii) discards all equilibria that are inefficient for the voters of any winning coalition.

\subsection{Endogenous Aggregation Games}

Endogenous aggregation games consist of two phases:
\begin{itemize}
\item A {\em pre-vote phase}, where, starting from a uniform aggregation game, players make simultaneous transfers of utility that may modify each others' payoffs in the game;
\item A {\em vote phase}, where players play the aggregation game resulting from the original game after payoffs are updated according to the tranfers made in the pre-vote phase.
\end{itemize} 
As usual, it is assumed that the players have common knowledge of the structure of the game (including their goals and payoffs). A key assumption is furthermore that the transfers made in the pre-vote phase be binding, for instance through a central authority.\footnote{The existence of such authority is a common assumption in work on election control and bribing \citep{faliszewski16control},  as well as persuasion \citep{HazonEtAlIJCAI2013}. It is often referred to as `chair' or `election organiser'.}  The authority would, besides running the election in the vote phase, also collect and enforce the transfers announced in the pre-vote phase. 

\medskip

Pre-vote strategies are modelled as {\bf transfer functions} of the form:
\[
\tau_i: \D^\N \times \N \to \mathbb{R}_+
\] 
where $i \in \N$.
These functions encode the amount of payoff that player $i$ commits to give to player $j$ should a given ballot profile $\B$ be played, in symbols, $\tau_i (\B,j)$. The set of all transfer functions is denoted by $\T$, and a {\bf transfer profile} is a tuple of transfer functions $\tau \in \T^{|\N|}$. We denote with $\tau^0$ the `void' transfer where at every profile every player gives $0$ to the other players.

The aggregation game induced by the transfer profile $\tau$ from $\A$ is denoted $\tau(\A) = \tuple{\N,\I, F, \set{\gamma_i}_{i\in\N}, \set{\tau(\pi)_i}_{i\in\N}}$ where, for any $i \in \N$:
\begin{equation}
\tau(\pi)_i(\B) = \pi_i(\B) + \sum_{j\in \N} \left( \tau_j (\B,i) - \tau_i (\B,j)\right) \label{eq:payoff}
\end{equation}
So the payoff of player $i$ at profile $\B$ {\em once $\tau$ is played}, consists of the old payoff that $i$ was receiving at $\B$, {\em plus} the money that $i$ receives from the other players at $\B$, {\em minus} what $i$ gives to them at $\B$. Notice that transfers do not preserve the uniformity of payoffs: even though $\A$ is always assumed to be uniform, $\tau(\A)$ is not necessarily so,\footnote{In fact it is easy to see that there always exists a transfer that turns a uniform aggregation game into a non uniform one.} and may lack a NE (recall Proposition \ref{prop:noexist}).

\medskip

Endogenous aggregation games are defined formally as follows:
\begin{definition}[Endogenous aggregation games] \label{def:JA2}
An endogenous aggregation game is a tuple
$\EA = \tuple{\J,\{T_i\}_{i\in N}}$ where $\J$ is a uniform aggregation game, and for each $i \in \N$, $T_i= \T$. A {\em constant} endogenous aggregation game is an endogenous aggregation game $\EA$ where $\A$ is assumed to be constant.\footnote{Constant endogenous aggregation games will be discussed in Section \ref{section:extendedgoals}.}
\end{definition} 
In an endogenous aggregation game, each player $i$ selects first a transfer $\tau_i$ (pre-vote phase). The resulting transfer profile $\tau$ yields the new aggregation game $\tau(\A)$. Then each player selects a ballot $B_i$ (vote phase) in $\tau(\A)$. The resulting ballot profile $\B$ yields the collective ballot $F(\B)$, where $F$ is the aggregator of $\A$.

\subsection{Equilibrium Selection via Pre-vote Negotiations}

In this section we first provide the definition of the solution concept we use for analysing endogenous games as two-stage extensive form games, and we then present our main results on the selection of efficient equilibria via pre-vote negotiations.

\subsubsection{Solving Endogenous Aggregation Games}

Endogenous aggregation games are (perfect information) extensive form games with two stages of simultaneous choices. 
So in an endogenous aggregation game, a strategy of player $i$ is a function $\sigma_i$ that selects a $\tau_i$ in the pre-vote phase, and a ballot $B_i$ in every aggregation game reachable in the voting phase of the endogenous game. A profile $\sigma$ of such strategies selects a unique transfer profile and a ballot profile---{\bf transfer-ballot} profile in short---$(\tau, \B)$ that will be played if that strategy profile is chosen by the players.
It is fundamental to notice how each strategy profile determines not only what will be played after a given transfer profile, i.e., the choices {\em on the equilibrium path}, but also what would be played after all other possible transfer profiles, i.e., the counterfactual choices {\em off the equilibrium path}.
Given the above, it would seem natural to resort to (pure strategy) subgame perfect Nash equilibrium (SPE) to solve endogenous aggregation games. 
There is a complication however. As we observed earlier, the aggregation game resulting from a transfer profile may not be uniform and, therefore, may not necessarily have a NE. This makes SPE inapplicable as some subgames of the initial extensive game may, therefore, be unsolvable.
Also notice how resorting to mixed strategy equilibria would not be a fix, as lexicographic preferences are well-known not to be representable in terms of a utility functions, therefore not enabling Nash's NE existence theorem \citep{ruby}.

What we do instead is to analyse endogenous aggregation games by assuming that, if the game resulting from a transfer profile does not have a NE, players play a maxmin strategy and resort to their maxmin value, or {\em security level},\footnote{See \citep{shoham08multiagent} \citep{shoham08multiagent} for an extensive exposition of the notion.} to evaluate the resulting aggregation game. \footnote{A game with no NE should be seen as an unstable game, i.e., a situation in which the players do not have any reason to believe that some specific outcome will be realised. Measures such as the security level therefore compensate for this uncertainty. We could have adopted a number of alternative solutions, e.g., taking the value of the minimum outcome for each player, or considering such games never to be profitable deviations. Ultimately, all these assumption have the effect of ruling out profitable deviations to games with no NE.} 
In our quasi-dichotomous setup a maxmin strategy (ballot) is defined as:
\begin{align}
\mathsf{argmax}^{\succeq^{\pi}_i}_{B_i} \mathsf{min}^{\succeq^{\pi}_i}_{\B_{-i}} \set{F(\B) \mid \B = (B_i,\B_{-i})},
\end{align}
i.e., the ballot that maximises the minimum, with respect to $\succeq^{\pi}_i$ (Definition \ref{def:pref}), outcome for $i$. If the maxmin strategy $B_i$ guarantees that $\gamma_i$ is satisfied no matter what the other players do, then we say that $B_i$ is {\bf safe for} $\gamma_i$. The minimum payoff guaranteed by a $i$'s maxmin strategy is $i$'s {\bf security level}. We can now define the solution for an endogenous aggregation game. 

\begin{definition}[Solutions] \label{def:sol}
Let $\EA$ be an endogenous aggregation game. A strategy profile $\sigma$ of $\EA$ is a {\em solution}  of $\EA$ if and only if:
\begin{enumerate}
\item {\em Voting phase.} At each $\tau(\A)$ with $\tau \in \T$ 
, $\sigma$ induces a NE of $\tau(\A)$, if such equilibrium exists, or a profile of maxmin strategies of $\tau(\A)$, otherwise.

\item {\em Negotiation phase.} $\sigma$ selects a unique transfer profile $\tau$, and this profile
is a NE of the game $\tuple{\N, \set{T_i}_{i \in \N}, \set{\succeq^{\pi}_i}_{i \in \N}}$ where $T_i = \T$ and each $\succeq^\pi_i \subseteq \T^2$ is the quasi-dichotomous preference\footnote{Recall Definition \ref{def:pref}.}over transfer profiles where $\tau \succeq^\pi_i \tau'$ if {\em either} the outcome of $\tau(\A)$ satisfies 
$i$'s goal and the outcome of $\tau'(\A)$ does not, {\em or} $i$'s payoff is greater at the outcome of $\tau(\A)$ than at the outcome of $\tau'(\A)$.\footnote{Formally:
a profile $\tau$ with associated ballot $\B$ by strategy profile $\sigma$ satisfies goal $\gamma_i$ ($\tau \models \gamma_i$) if {\em either} $\B$ is a NE of $\tau(\A)$ and $\B \models \gamma_i$, {\em or} $B_i$ is safe for $\gamma_i$ in $\tau(\A)$, where $\B$ is the ballot profile induced by $\sigma$, and $B_i$ is $i$'s maxmin strategy. The payoff $\pi_i$ of profile $\tau$ is {\em either} $\tau(\pi)_i(\B)$ if $\B$ is a NE, {\em or} $i$'s security level in $\tau(\A)$ otherwise, where $\B$ is the ballot profile induced by $\sigma$.  
}
\end{enumerate}
A transfer-ballot profile $(\tau,\B)$ is called a {\bf solution outcome} of $\EA$ iff it is induced by some solution of $\EA$.
\end{definition}
Intuitively, a solution outcome is obtained as follows: 
\begin{itemize}

\item {\em Voting phase} In the aggregation game resulting after each transfer profile, a NE is selected where at least one such equilibrium exists, or a maxmin profile is selected otherwise; the value of such NE, if it exists, or the players' security levels are used to evaluate transfer profiles in the next phase;
\item {\em Negotiation phase}
A transfer profile is selected, such that no profitable deviation exists to another transfer profile, for any player, given the selected continuations.

\end{itemize}

The definition of solutions is an adaptation to endogenous aggregation games of the classical backwards induction procedure used to calculate subgame perfect equilibria in extensive games, and therefore gives a procedural recipe to compute them.


\subsubsection{Surviving Equilibria}

This section provides existence results for solutions of endogenous games. We actually provide stronger results, showing not only that solutions exist, but also (in the next sections) necessary and sufficient conditions for them to enjoy desirable properties in terms of the `quality' of the equilibria that negotiation gives rise to. In other words, we study endogenous aggregation games as endogenous mechanisms for the refinement of the equilibria of their underlying aggregation games. 
\begin{definition}
Let $\A$ be a uniform aggregation game, and $\prof{B}$ a NE of $\A$. $\prof{B}$ is a {\em surviving Nash equilibrium} (SNE) of $\A$ if there exists a transfer profile $\tau$ such that $(\tau, \B)$ is a solution outcome of the endogenous aggregation game $\EA$.
\end{definition}
Intuitively, surviving SNE identify those voting outcomes that can be rationally sustained by an appropriate pre-vote negotiation. Clearly, not all NE of the initial game will be surviving equilibria but, crucially, it can be shown that NE with good properties are of this kind, as we set out to show in this section. Let us first establish the following lemma.
\begin{lemma}\label{lemma:tau}
Let $\A$ be a uniform aggregation game for an independent and monotonic aggregator $F$. 
Then, for every $\N$-efficient and $\N$-truthful ballot profile $\prof{B}$ of $\J$, there exists a transfer profile $\tau$ such that $\prof{B}$ is a weakly dominant strategy equilibrium in $\tau(\A)$.
\end{lemma}
\begin{proof}
Let $\prof{B}$ be a $\N$-efficient and $\N$-truthful ballot profile of $\J$. 
We construct a transfer profile $\tau$ such that for any $i \in \N$ the strategy $B_i$ is a strictly dominant strategy for $i$ in $\tau(\A)$, that is: for any profile $\prof{B}'$, $(B_i,\prof{B}'_{-i}) \succeq^\pi_i\prof{B}'$ and for some $\prof{B}'$, $(B_i,\prof{B}'_{-1}) \succ^\pi_i\prof{B}'$. The construction goes as follows.\footnote{The construction is based on a similar construction developed in \cite[Th. 4]{JW05}.}
Consider now the quantity:
\begin{equation}
M = 1 + \max \left(\set{z \mid \exists \prof{B},\prof{B}'  \mathit{and}~  i\in \N \mathit{s.t.}~  z = \pi_i(\prof{B}) - \pi_i(\prof{B}')}\right) \label{eq:M}
\end{equation}
that is, $M$ exceeds by one unit the maximal difference between the payoff received at any two outcomes by any agent.
We are now ready to define a transfer function. For all $i,j \in \N$: 
\begin{equation}
\tau_i(\prof{B}',j) = \left\{
\begin{array}{ll}
2M & \mathit{if}~  B'_{i} \neq B_i \\
0 & \mathit{otherwise}. 
\end{array}
\right.
\label{formula:tau}
\end{equation}
In words, each player $i$ commits to pay each other player the sum $2M$ in case he deviates from the ballot $B_i$. 
By the definition of the preference relation $\succeq^\pi_i$ (Definition \ref{def:pref}), in order to prove the claim we need to show that: 
\fbox{A} For any $\prof{B}'$ if $F(\prof{B}')\models \gamma_i$ then $F(B_i,\prof{B}'_{-i}) \models \gamma_i$, that is, under no circumstance $i$'s goal would become falsified by playing $B_i$.
\fbox{B} For any $\prof{B}'$ if $F(\prof{B}')\models \gamma_i$ iff $F(B_i,\prof{B}'_{-i}) \models \gamma_i$ (that is both profiles satisfy $i$'s goals), then $\tau(\pi)_i(B_i,\prof{B}'_{-i}) \geq \tau(\pi)_i(\prof{B}')$, and for some $\prof{B}'$, $\tau(\pi)_i(B_i,\prof{B}'_{-i}) > \tau(\pi)_i(\prof{B}')$. 

\fbox{Claim A} 
Assume that $F(\prof{B}')\models \gamma_i$. There are two cases. First, $B_i' \models \gamma_i$. Since $\B$ is $\N$-truthful we know that $B_i \models \gamma_i$. Now $\gamma_i$ is assumed to be a cube (Definition \ref{def:JA1}) so $B'_i$ and $B_i$ can differ only on the evaluation of issues on which the truth of $\gamma_i$ does not depend. By the independence of $F$ we conclude that $F(B_i,\prof{B}'_{-i}) \models \gamma_i$. Second, $\B' \not\models \gamma_i$. But, as above, $B_i \models \gamma_i$. By the assumption $F(\prof{B}')\models \gamma_i$, so in profile $(B_i, \B'_{-i})$ voter $i$ is changing the evaluation of at least one of the atoms that was falsifying $\gamma_i$ in $B'_i$ (recall that $\gamma_i$ is a cube) to the evaluation given by $F(\B')$. By the monotonicity of $F$ we can therefore conclude that $F(B_i,\prof{B}'_{-i}) \models \gamma_i$.

\fbox{Claim B} Given the definition of payoffs in endogenous games \eqref{eq:payoff}, the claim follows directly from the construction of $\tau$. This completes the proof.
\end{proof}
Intuitively, the lemma establishes that whenever a truthful profile exists that satisfies all players' goals, that profile can be turned into a dominant strategy equilibrium by means of a suitable combination of pre-vote transfers which, in essence, make the players' commitment to that equilibrium credible. 


\subsubsection{Good Equilibria Survive}

Lemma~\ref{lemma:tau}  establishes the existence of a suitable transfer profile sustaining `good'---that is, truthful and efficient---equilibria. 
We now show that there exist solutions of the endogenous aggregation game which require precisely the transfer profile from Lemma~\ref{lemma:tau}  to be played in the negotiation phase of the game, and which therefore determine the voters to play `good' equilibria in the voting phase.
\begin{theorem}\label{theorem:survival} 
Let $\A$ be 
a uniform aggregation game for an independent and monotonic aggregator $F$. 
Every $\N$-efficient and $\N$-truthful NE of $\J$ is a SNE.
\end{theorem}
\begin{proof}
Let $\prof{B}$ be a $\N$-efficient and $\N$-truthful NE of $\J$. By the assumptions on the aggregator and Lemma \ref{lemma:tau}, we can construct a transfer profile $\tau$ according to Formula \eqref{formula:tau} such that $\prof{B}$ is a weakly dominant strategy equilibrium of $\tau(\A)$. 
We have to show that  $(\tau, \prof{B})$ is a solution outcome of $\EA$. Suppose towards a contradiction that this is not the case, that is, that there exists a profitable deviation by player $i$ from the strategy profile that induces the transfer-ballot profile $(\tau, \prof{B})$. Since $\B$ is a dominant strategy equilibrium in $\tau(\A)$ such deviation has to involve a different transfer $\tau^*_{i}$ by player $i$ in the pre-vote phase. We identify two cases:

\fbox{Case 1} The deviation induces a new transfer profile $\tau'=(\tau^*_{i},\tau_{-i})$ such that $\tau'(\A)$ has no NE.
By Definition \ref{def:sol} a deviation to $\tau^*_{i}$ would yield $i$ his security level in $\tau'(\A)$. We show that $\tau \succeq^\pi_i \tau'$, and hence the deviation cannot be profitable for $i$. Since $\B \models \gamma_i$ we have, by Definition \ref{def:sol}, that if $\tau' \models \gamma_i$
then $\tau \models \gamma_i$, that is, it cannot be the case that $\tau'$ leads $i$ to satisfy its goal while $\tau$ does not. As to the payoffs, observe that $\pi_i(\tau)$ is the payoff yielded by the NE $\B$ of $\tau(\A)$, and that given the construction of $\tau'$ no agent transfers any utility to $i$ in $\tau'$: for all $j \in \N \setminus \{i\}$, $\tau'_{j}(\prof{B}',i)=0$ whenever $B'_{j} = B_j$. It follows that for every ballot profile $\B'$ in $\tau'(\A)$ $i$'s payoff is at most as much as what it obtains from $\B$ in $\tau(\A)$. Hence, $\tau \succeq^\pi_i \tau'$.

\fbox{Case 2} The deviation induces a new transfer profile $\tau'=(\tau^*_{i},\tau_{-i})$ such that the game $\tau'(\A)$ has a NE. We show that $\tau^*_{i}$ cannot be a profitable deviation for $i$.
Since the strategy profile supporting $(\tau, \B)$ is assumed to be a solution, $i$'s deviation determines a transfer-ballot profile $(\tau', \B')$ where $\B'$ may or may not be a NE. However, if $i$'s deviation toward $(\tau', \B')$ where $\B'$ is not a NE is profitable, $i$ must have another at least as profitable deviation determining the transfer-ballot profile $(\tau', (B^*_i, \B'))$, where $(B^*_i, \B')$ is a NE of $\tau'(\A)$. We therefore assume w.l.o.g., that the transfer-ballot profile $(\tau', \B')$ determined by $i$'s deviation is such that $\B'$ is a NE of $\tau(\A)$.  Observe now that in order for $(\tau', \prof{B}')$ to be profitable for $i$ there must exist a $j \neq i$ such that $B'_j \neq B$, that is, there exists at least another agent $j$ besides $i$ who deviates from $\B$ to $\prof{B}'$. This is because $\prof{B}$ is an $\N$ efficient NE and    by construction of $\tau'$,we have that $\tau'_{j}(\prof{B}',i)=0$ whenever $B'_{j} = B$. So let there be $k \geq 1$ players  $j \neq i$ for which $B'_{j} \neq B$ and consider some such $j$. We have two subcases:
 
 \fbox{Case 2a} $\prof{B}'\not\models \gamma_j$. Note that since $\prof{B}'$ is a NE, we have that no $B''_j$ is such that $F({B}''_{j},\prof{B}'_{-j}) \models \gamma_j$, so $B'_j$ is a best response only by virtue of $j$'s payoff. By the construction of $\tau'$, playing $B'_j$ gives $j$ the following payoff: 
\[
\tau'(\pi)_j(\B') = \pi_j(\B')- (|\N|-1)2M + 2M(k-1)+\tau'_{i}(\B',j). \label{eq:tau1}
\]
If $j$ plays $B_j$ instead, then $j$'s payoff is: 
\[
\tau'(\pi)_j(B_j, \B'_{-j}) = \pi_j(B_j,\prof{B}'_{-j}) + 2M(k-1)+\tau'_{i}(B_j,\B'_{-j},j). \label{eq:tau2}
\]

Now since $\prof{B}'$ is a NE by assumption and $j$ does not have a better response that can satisfy her goal, we have that $\tau'(\pi)_j(\B') \geq \tau'(\pi)_j(B_j, \B'_{-j})$ and therefore: 
\[
\tau'_{i}(\prof{B}',j) - \tau'_{i}(B_j,\prof{B}'_{-j},j) \geq \pi_j(B_j,\prof{B}'_{-j}) - \pi_j(\prof{B}') + (|\N|-1)2M.
\]
Now we use this inequality to compare $i$'s payoff in $(\tau, \prof{B})$ vs. $(\tau', \prof{B}')$.
Given the definition of $M$ in Formula \eqref{eq:M} and given the fact that $|\N|-1\geq 2$ it follows that $\tau'_{i}(\prof{B}',j) - \tau'_{i}(B_j,\prof{B}'_{-j},j) \geq 3M$, which
in turn implies that $\tau'_{i}(\prof{B}',j) \geq 3M$. 
Therefore $i$'s payoff $\tau'(\pi)_i(\prof{B}')$ in the new NE $\prof{B}'$ is at most:
\[
\pi_i(\prof{B}') - k3M + k2M.
\]
In contrast, by the construction of $\tau$ in Formula \eqref{formula:tau}, $\tau(\pi)_i(\prof{B}) = \pi_i(\prof{B})$.
So from the fact that $k \geq 1$ it follows that $\pi_{i}(\prof{B}') - k3M + k2M \leq \pi_{i}(\prof{B})$ and therefore that $\tau'(\pi)_i(\prof{B}') \leq \tau(\pi)_i(\prof{B})$. Since $\prof{B}$ is $\N$-efficient, by Definition \ref{def:sol}, $\tau \succeq^\pi_i \tau'$ and the constructed deviation $\tau^*_i$ cannot be profitable.

 \fbox{Case 2b} $F(\prof{B}') \models \gamma_j$. Notice that from this it follows that $F(B_j,\prof{B}'_{-j}) \models \gamma_j$, for otherwise there would be a ballot profile (i.e., $(B_j,\prof{B}'_{-j})$) where $B'_j$ is a better response for $j$ making $\gamma_j$ true. From this we would conclude that $\B$ in $\tau(\A)$ would not be a dominant strategy equilibrium, against our assumption.
We can then proceed as in Case 2a, showing that  $\tau'(\pi)_i(\prof{B}') \leq \tau(\pi)_i(\prof{B})$. Since, as just noticed, both $\B$ and $\B'$ satisfy $\gamma_j$, by Definition \ref{def:sol}, $\tau \succeq^\pi_i \tau'$ and the constructed deviation $\tau^*_i$ cannot be profitable. This completes the proof.
\end{proof}

\noindent
By combining~Theorem~\ref{theorem:survival} with Proposition~\ref{prop:NE-existence}, we also get the following corollary:

\begin{corollary}
Let $F$ be an independent and monotonic aggregator. Every $\N$-consistent uniform aggregation game for $F$, where $\N \in \Winr_{\bigwedge \Gamma}$ with $\Gamma = \set{\gamma_i \mid i\in \N}$,  has a SNE that is $\N$-truthful and $\N$-efficient.
\end{corollary}


\subsubsection{Only Good Equilibria Survive}

We now turn to necessary conditions for equilibria to survive, and observe that for an equilibrium to survive, it has to be efficient. This can be regarded as a converse statement of Theorem~\ref{theorem:survival}, and it holds true in a general form concerning winning coalitions (and not just the coalition $\N$) with internally consistent goals:

\begin{theorem}\label{prop:surviving}
Let $\A$ be a $C$-consistent uniform aggregation game for an independent and monotonic aggregator $F$, and such that $C \in \Win_{\bigwedge \Gamma}$ where $\Gamma = \set{\gamma_i \mid i\in C}$.  Then, every SNE of $\A$ is $C$-efficient. 
\end{theorem}

\begin{proof}[Proof]
We proceed by contraposition.
Let $\prof{B}$ be a NE that is not $C$-efficient, i.e., such that $F(\prof{B})\not \models \gamma_i$ for some individual $i \in C$, and assume towards a contradiction that $\prof{B}$ is a SNE.
Therefore there exists a transfer profile $\tau$ such that  $(\tau, \prof{B})$ is a solution outcome of the endogenous game $\EA$. We proceed towards a contradiction and construct a profitable deviation for a player $i$ from $\tau$, that is a $\tau'$ such that $\tau' \succ^\pi_i \tau$.
By $C$-consistency of $\J$ there exists a ballot $B'$ such that $B'\models \bigwedge \Gamma$, hence in particular $B'\models \gamma_i$. Let now $i$ deviate to any transfer profile $\tau'=(\tau'_i, \tau_{-i} )$ such that she offers to all other players in $C$ more than their payoff difference if they switch to vote for ballot $B'$ while everybody else in $C$ does so, i.e., 
\[
\tau'_i((\B'_C, \B''_{-C}),j)  >  \tau(\pi)_j(\B'_{C-\set{j}}, B''_j, \B''_{-C}) - \tau(\pi)_j(\B'_C, \B''_{-C})
\]
for each $j\in C$, each $B''_j$, and each $\B''_{-C}$, and where $\B'_C = (B'_j)_{j \in C}$ . 
By the fact that $B'$ is $C$-efficient, $F$ independent and monotonic, and $C$ is a winning coalition for $\bigwedge \Gamma$, this transfer makes each $B'_j$, with $j \in C$ a best response in profiles $(\B'_C, \B''_{-C})$, for any $\B''_{-C}$. It follows that each $(\B'_C, \B''_{-C})$ is a NE of $\tau'(\A)$, which satisfies $\bigwedge \Gamma$, while $\B$ in $\tau(\A)$ does not. We can therefore conclude (by Definition \ref{def:sol}) that $\tau' \succ^\pi_i \tau$, completing the proof.
\end{proof}
The proof of Theorem~\ref{theorem:survival} provides every agent with a simple algorithm to compute a negotiation strategy that will guarantee the emergence of an efficient equilibrium in the resulting game. Note that the use of cubes sidesteps the intractability of the satisfiability problem for the conjunction of the goals.


Observe also that Theorem~\ref{prop:surviving} implies the existence of uniform aggregation games where no equilibrium is surviving, which in turn implies that Theorem \ref{theorem:survival} cannot be weakened from $\N$-efficiency to $C$-efficiency. This is the case when distinct but overlapping coalitions have incompatible goals, as the following example shows:

\begin{example}\label{ex:coalitions}
Let there be five players in $\N$, and let $F$ be the majority rule. Let $\gamma_1=p \wedge \neg r$, $\gamma_2=\gamma_3=\gamma_4=\top$ and let $\gamma_5=r \wedge \neg p$. Both coalitions $C_1=\{1,2,3,4\}$ and $C_2=\{2,3,4,5\}$ are resilient winning coalitions, and the game is both $C_1$-consistent and $C_2$-consistent. Hence, by Theorem~\ref{prop:surviving} any surviving equilibria must be both $C_1$-efficient and $C_2$ efficient, which is impossible given that the conjunction of the goals of the two coalitions is inconsistent.
\end{example}


\subsubsection{Summary}

Results such as Theorems~\ref{theorem:survival} and \ref{prop:surviving} suggest that pre-vote negotiations are a powerful tool players have to overcome the inefficiencies of aggregation rules. More specifically, when the goals of all players can be satisfied at the same time, pre-vote negotiations allow players to engineer side-payments---essentially as devices for credible commitments---leading to equilibrium outcomes that satisfy them, and ruling out all the others. 
We stress that in solving endogenous aggregation games, even when the game ends up sustaining efficient outcomes, players' strategies are individually rational and the game remains non-cooperative throughout. This differentiates our work from approaches to equilibrium selection with coalitional games, such as the one developed by \citep{bachrach11coalitional}.
%



\section{Larger Classes of Goals}\label{section:extendedgoals}

In this section we explore two variants of our results where we relax the assumption that goals need to be cubes. The section shows that our results about equilibrium refinement via endogenous games are robust across different natural settings to model strategic behaviour in binary aggregation.


\subsection{Monotonic Goals}

As shown in Examples \ref{example:mono} and \ref{example:cubes}, when goals are not cubes it is possible to construct constant aggregation games in which truthful strategies are not dominant. In this section we find sufficient conditions on the class of goals for which there exists a truthful strategy that is weakly dominant, and we build further on this result by showing that, in constant aggregation games where players' goals are `aligned' in a precise sense, we can still construct a truthful NE that is also efficient (Theorem \ref{coro:survival}).

Up until now we referred to aggregation games assuming players' goals were cubes (Definition \ref{def:JA1}). In this section we drop such constraint and, when referring to aggregation games, we will make explicit what type of goal formulas we are considering.


We first need some additional notation and a definition. Let $B$ be a valuation over $\PS$, denote with $B_{-j}$ the restriction of $B$ to $\PS \setminus \{p_j\}$. 

\begin{definition} A formula $\phi \in \L_\PS$ is:
\begin{description}
\item{{\bf {\em positively monotonic}}} in $p_j$, where $p_j$ occurs in $\phi$, if for all $B, B' \in \D$ such that $B\models \phi$, $B_{-j}=B'_{-j}$, and $B' \models p_j$ then $B'\models \phi$.   
\item{{\bf {\em negatively monotonic}}} in $p_j$ if for all $B, B' \in \D$ such that $B\models \phi$, $B_{-j}=B'_{-j}$, and $B' \not\models p_j$ then $B'\models \phi$. 
\end{description}
A formula $\phi$ is {\bf {\em monotonic}} if for every $p_j$ occurring in $\phi$, it is either positively or negatively monotonic in $p_j$.
\end{definition}
Let us mention a few examples. Formulas that are satisfied by only one model are monotonic. Similarly, cubes (i.e., conjunctions of literals) are monotonic. Not all monotonic formulas are, however, cubes: disjunctions of literals, which were used to construct Example~\ref{example:mono}, are monotonic but are not cubes in general. Again, Example~\ref{example:cubes} provides a good example of formulas which are not monotonic in the above sense.

\medskip

We have shown (Proposition \ref{lemma:weak}) that when goals are cubes, every truthful strategy in a constant aggregation game is weakly dominant.
Although this fails to be true when goals are monotonic (Example \ref{example:mono} offers a counterexample to such claim), we can still establish the following weaker result:
\begin{lemma}\label{prop:dominant}
Let $\A$ be a constant aggregation game with monotonic goals for an independent and monotonic rule $F$. Then, for each player $i \in \N$ there exists a truthful strategy that is weakly dominant in $\A$. 
\end{lemma}
\begin{proof}
The proof is by construction. Recall Definition \ref{def:pref}. We build a ballot $B$ such that $B\models\gamma_i$ and show that for any ballot $B'\not =B$ and profile $\B$, we have that if $F(\B_{-i},B'_i)\models\gamma_i$ then $F(\B_{-i},B_i)\models\gamma_i$.
Consider a ballot $B''$ such that $B'' \models\gamma_i$ (such ballot exists as we assume goals to be consistent). Now build the ballot $B$ as follows. For any issue $p_j$, let:
\[
B(j) = 
\left\{
\begin{array}{ll}
1 & \mbox{if $\gamma_i$ is positively monotonic in $j$}\\
0 & \mbox{if $\gamma_i$ is negatively monotonic in $j$}\\
B''(j) & \mbox{otherwise}
\end{array}
\right.
\]
By the assumption of monotonicity of $\gamma_i$ the above construction guarantees that $B \models \gamma_i$. 
We now show that for any $B'\not = B$ we have that if $F(\B_{-i},B'_i)\models\gamma_i$ then $F(\B_{-i},B_i)\models\gamma_i$ as well, for any ballot profile $\B$, thereby establishing the claim. So take such a ballot $B'$ and assume that $F(\B_{-i},B')\models\gamma_i$. Now let $n$ be the Hamming distance\footnote{We recall that the hamming distance $H(B,B')$ between two valuations (vectors) $B$ and $B'$ is defined as follows: $H(B,B')=\sum_j |B(j)-B'(j)|$. \label{foot:ham}}
between $B'$ and $B$ and consider the sequence $B' = B^0, \ldots, B^n = B$ such that for all $1 \leq k \leq n$, $B^{k +1}$ is equal to $B^k$ except for the value of one of the issues (that is, the sequence consisting of a minimal number of value swaps to turn $B'$ into $B$). 
By construction of $B$, each $B^k$ is the result of swaps $0 \mapsto 1$ for issues in which $\gamma_i$ is positively monotonic, and $1 \mapsto 0$ for issues in which $\gamma_i$ is negatively monotonic,
and possibly other swaps on variables not occurring in~$\gamma_i$.

We now prove that
$F(\B_{-i},B^k)\models\gamma_i$ for all $0 \leq k \leq n$. We proceed by induction over $k$. 
For $k = 0$ the claim is true by assumption.  
Assume the claim is true for $k$ (IH), 
i.e., $F(\B_{-i},B^k)\models\gamma_i$. 
By the observation above, $B^{k+1}$ can be obtained from $B^k$ in only two ways: \fbox{a} by a swap $0 \mapsto 1$ for some issue $j$ in which $\gamma_i$ is positively monotonic; or \fbox{b} by a swap $1 \mapsto 0$ for some issue $j$ in which $\gamma_i$ is negatively monotonic. 
If \fbox{a} is the case
then in profile $(\B_{-i},B^{k+1})$ there is one more voter, namely $i$, who supports $p_j$. 
We distiguish two cases: 
\begin{description}

\item \fbox{a1} If $F(\B_{-i},B^k)\models p_j$, by independence and monotonicity of $F$, we can infer that also $F(\B_{-i},B^{k+1})\models p_j$ and therefore $F(\B_{-i},B^{k+1}) = F(\B_{-i},B^{k})\models\gamma_i$ by IH. 

\item \fbox{a2} If $F(\B_{-i},B^k)\not\models p_j$, then either $i$ is not pivotal on issue $j$, and therefore $F(\B_{-i},B^{k+1}) = F(\B_{-i},B^{k}) \models\gamma_i$ by IH. Or $i$ is pivotal at step $k$, and then $F(\B_{-i},B^{k+1})\models p_j$. 
We now need to observe that by the assumption we have that $\gamma_i$ is positively monotonic in $p_j$ to conclude that, also in this case, $\maj(\B_{-i},B^{k+1})\models\gamma_i$. 
\end{description}

\noindent
This proves the claim for the first case. 
If \fbox{b} is the case, we reason in a symmetric fashion to conclude that that $F(\B_{-i},B)\models\gamma_i$, as claimed.
\end{proof}


A direct consequence of Lemma~\ref{prop:dominant} is the existence of a truthful NE in weakly dominant strategies for constant aggregation games with monotonic goals. However, in order to obtain results analogous to Theorem~\ref{theorem:survival} in the context of monotonic goals we need first to establish sufficient conditions for the existence of truthful {\em and} efficient NE. 

\smallskip

Let us first introduce some auxiliary terminology. Two monotonic goals $\gamma_1$ and $\gamma_2$ are {\bf aligned} if they are positively monotonic on the same set of issues and negatively monotonic on the same set of issues. We establish the following result:

\begin{lemma}\label{coro:constant}
Let $\A$ be a constant aggregation game with monotonic goals, for an independent and monotonic aggregator $F$. Assume moreover that $\A$ is $C$-consistent, that $C\in \Win_{\bigwedge\Gamma}$ where $\Gamma = \set{\gamma_i \mid i\in C}$, and that all goals of agents $i\in C$ are aligned. Then, $\A$ has a truthful and $C$-efficient NE (in weakly dominant strategies).
\end{lemma}

\begin{proof}
Let $B''$ be a ballot such that $B''\models \bigwedge_{i\in C} \gamma_i$. $B''$ exists by assumption of $C$-consistency. $B''$ can be used in the proof of Lemma~\ref{prop:dominant} to construct a truthful NE in weakly dominant strategies. 
Now observe that, since all goals of agents in $C$ are aligned, the weakly dominant equilibrium so constructed is composed by a unanimous ballot choice for individuals in $C$, which we shall call $B^*$. 
Since $C$ is a winning coalition for $\bigwedge_{i\in C} \gamma_i$, we have that 
$F(\B)\models \bigwedge_{i\in C} \gamma_i$ and the equilibrium is therefore also $C$-efficient as claimed.
\end{proof}

Everything is now in place to prove a variant of Theorem \ref{theorem:survival} for constant aggregation games where players hold monotonic goals.

\begin{theorem}\label{coro:survival} 
Let $\A$ be an $\N$-consistent aggregation game with monotonic goals, for an independent and monotonic aggregator $F$. Assume that all individual goals are aligned, and that $N \in \Win_{\bigwedge \Gamma}$ for $\Gamma = \set{\gamma_i \mid i\in \N}$. Then, there exists an $\N$-truthful and $\N$-efficient NE of $\A$ which is also a SNE of $\A$.
\end{theorem}

\begin{proof}[Sketch of proof]
Let $\B$ be an $\N$-truthful and $\N$-efficient NE of $\A$, which exists by Lemma \ref{coro:constant}. We can therefore construct a transfer profile $\tau$ such that $\B$ is a weakly dominant strategy equilibrium of $\tau(\EA)$. The construction proceeds in the same way as the construction, for a transfer profile of the same type, given in the proof of Lemma \ref{lemma:tau}.
We then have to show that $\B$ is a SNE of $\EA$. To do so we proceed towards a contradiction and suppose that there exists a profitable deviation $\tau^*_i$ for some player $i$ in $\EA$. The argument used in the proof of Theorem \ref{theorem:survival} to the effect that no such profitable deviation exists can be applied directly, thereby establishing the claim.
\end{proof}

The assumption of monotonicity cannot be further weakened. If goals are allowed to be non-monotonic, then it is possible to construct (constant) aggregation games where, for some player, no truthful strategy is weakly dominant as witnessed by Example~\ref{example:mono}.

\subsection{Arbitrary Goals}

As a final generalisation of our results, we consider the question of whether pre-vote negotiations can support some good equilibria once we lift any restriction on the logical structures of players' goals. We answer this question positively, but only at the cost of restricting the result to a specific aggregator, issue-wise majority.


\begin{theorem}\label{theorem:arbitrary} 
Let $\A$ be a uniform $\N$-consistent aggregation game for $\maj$ with arbitrary goals in $\L_\PS$.
There exists a $\N$-efficient and $\N$-truthful NE of $\A$ that is a SNE of $\A$.
\end{theorem}
\begin{proof}
Let $\prof{B}$ be a $\N$-efficient and $\N$-truthful NE of $\A$ such  $B_i=B_j=B$ for each $i,j\in \N$ where $B \models \bigwedge_{i\in N} \gamma_i$. 
We construct a transfer function $\tau$ such that $(\tau, \B)$ is a solution outcome of $\EA$. Such equilibrium exists by the assumption of $\N$-consistency. Let $M$ be as in \eqref{eq:M}. 
For all $i,j \in \N$, let $\tau_i(\prof{B}',j) = 2Mx$, whenever $B'_j\neq B$,  where $x$ is the Hamming Distance of ballot $B'_j$ w.r.t. $B$, and $\tau_i(\prof{B}',j) = 0$, whenever $B'_i= B$.  In words, each player $i$ is committing to pay to the other players the quantity $2xM$, with $x$ increasing the further $i$'s ballot is from $B$.  We show that  $(\tau, \prof{B})$ is a solution outcome of $\EA$.
Suppose towards a contradiction that this is not the case, that is, that there exists a profitable deviation by player $i$ from the strategy profile that induces the transfer-ballot profile $(\tau, \prof{B})$. Since $\B$ is a NE in $\tau(\A)$ such deviation has to involve a different transfer $\tau^*_{i}$ by player $i$ in the pre-vote phase. We identify two cases:

\fbox{Case 1} The game $\tau'(\A)$, where $\tau'=(\tau^*_{i},\tau_{-i})$, has no associated NE. We reason exactly like in Case 1 of the proof of Theorem \ref{theorem:survival}. The argument shows that for every ballot profile $\B'$ in $\tau'(\A)$, where $\tau'=(\tau^*_i, \tau_i)$ $i$'s payoff is at most as much as what it obtains from $\B$ in $\tau(\A)$ and, therefore, that the deviation to $\tau^*_i$ is not profitable.

\fbox{Case 2} The game $\tau'(\A)$, where $\tau'=(\tau^*_{i},\tau_{-i})$, has a NE.
Since the strategy profile supporting $(\tau, \B)$ in $\EA$ is assumed to be a solution, $i$'s deviation determines a transfer-ballot profile $(\tau', \B')$ where $\B'$ may or may not be a NE. However, if $i$'s deviation toward $(\tau', \B')$, where $\B'$ is not a NE, is profitable, $i$ must have another at least as profitable deviation determining the transfer-ballot profile $(\tau', (B^*_i, \B'))$, where $(B^*_i, \B')$ is a NE of $\tau'(\A)$. We therefore assume w.l.o.g., that the transfer-ballot profile $(\tau', \B')$ determined by $i$'s deviation is such that $\B'$ is a NE of $\tau(\A)$. Observe now that in order for $(\tau', \prof{B}')$ to be profitable for $i$ there must exist a $j \neq i$ such that $B'_j \neq B_j$. That is, there exists at least another agent $j$ besides $i$ who deviates from ${B}_j=B$ in $\prof{B}'$. This is because $\prof{B}$ is a $\N$-efficient NE by assumption and by construction of $\tau'$, we have that $\tau'_{j}(\prof{B}',i)=0$ whenever $B'_{j} = B$. So let there be $k \geq 1$ players  $j \neq i$ for which $B'_{j} \neq B$. Notice also that $k>\frac{n}{2}$ else $\maj(\prof{B}')=B$. Consider now some such $j$.
We have three sub-cases:
 
 \fbox{Case 2a} There is a $j\in N$ such that $j\neq i$, $B'_j \neq B_j$ and $\maj(\prof{B}')\not\models \gamma_j$. Since $\prof{B}'$ is a NE, we have that no $B''_j$ is such that $\maj({B}''_{j},\prof{B}'_{-j}) \models \gamma_j$, so $B'_j$ is a best response only by virtue of $j$'s payoff. W.l.o.g. we can assume that the Hamming distance between $B'_j$ and $B$ is $1$, that is, all voters are voting as close as possible to $B$, therefore incurring the lowest possible non-$0$ transfers towards the other voters. 
We can now reason exactly like in Case 2a  of the proof of Theorem \ref{theorem:survival} concluding that the constructed deviation $\tau^{*}$ cannot be profitable.
 
%
%
%


\fbox{Case 2b} There is a $j\in \N$ such that $j\neq i$, $B'_j \neq B_j$,  $\maj(\prof{B}') \models \gamma_j \mbox{ and }$ for which there is $B^{\circ}_j$, with $H(B,B^{\circ}_j)< H(B,B'_j)$, such that: $\maj(B^{\circ}_j,\prof{B}'_{-j}) \models \gamma_j$. Ballot $B'_j$ is therefore a best response for $j$ only by virtue of her payoff. The argument used in 2a  of the proof of Theorem \ref{theorem:survival} applies again to establish that $\tau'(\pi)_i(\prof{B}') \leq \tau(\pi)_i(\prof{B})$, and therefore that the constructed deviation $\tau^*_i$ cannot be profitable.

\fbox{Case 2c} For all $j\in \N$ with $j\neq i$ and such that  $B'_j \neq B_j$ we have that  $\maj(\prof{B}') \models \gamma_j \mbox{ and }$ there is no $B^{\circ}_j$, with $H(B,B^{\circ}_j)< H(B,B'_j)$, such that $F(B^{\circ}_j,\prof{B}'_{-j}) \models \gamma_j$. Consider now an arbitrary issue $k$. We have two cases. Either 
$\maj(\prof{B}')(k) = B(k)$ or not. 
\fbox{Case 2c(i)} We have that $\maj(\prof{B}')(k) = B(k)$. But then, for each such $j$, it must be the case that $B'_j(k) = B_j(k)$, because we have that $\maj(\prof{B}') \models \gamma_j \mbox{ and }$ there is no $B^{\circ}_j$, with $H(B,B^{\circ}_j)< H(B,B'_j)$, such that $F(B^{\circ}_j,\prof{B}'_{-j}) \models \gamma_j$ and ballot $B'_j$ is therefore a best response for $j$ only by virtue of her payoff, so the argument used for Case 2a  of the proof of Theorem \ref{theorem:survival} can be applied. 
\fbox{Case 2c(ii)} We have that $\maj(\prof{B}')(k) \neq B(k)$. Consider some $j\neq i$ for which $B'_j(k) \neq B(k)$, which exists by the assumptions.
Again by the assumptions above we have that $j$ cannot choose ballot $B^{+}$, such that $B'_j(k')= B^{+}_j(k')$, for each $k'\neq k$ and $B^{+}_j(k) = B_j(k)$, without compromising her goal.  This means that in $\prof{B'}$, $j$ is pivotal for $k$, i.e.,  $\maj(B^{+}_j,\prof{B}'_{-j})(k) =  B(k)$. Also notice that, being $k$ arbitrary, this is true for each issue $k$ such that $B'_j(k)\neq B_j(k)$. But, because of case 2c(i), this means that $\maj(B^{+}_j,\prof{B}'_{-j}) = \maj(\prof{B})$ and $\maj(\prof{B})\models \gamma_j$.  However this contradicts the assumptions behind Case 2c, showing that 2c(ii) is impossible. But then, given the reasoning in Case 2c(i), it must be the case that $\prof{B} = \prof{B}'$, contradicting the assumptions, and therefore showing that 2c altogether is impossible.

We conclude that in all cases the constructed deviation $\tau^*_i$ is not profitable, thereby establishing the claim.
\end{proof}
Let us comment on the above result in the light of the findings of this section. While Theorems~\ref{theorem:survival} and \ref{prop:surviving} showed that with cubes pre-vote negotiations exactly sustain the set of good equilibria, weakening the constraints imposed on players' goals reduces this set significantly. Theorem~\ref{coro:survival} shows that when goals are monotonic at least one good Nash equilibrium survives, provided the aggregator is independent and monotonic and the game constant. Theorem~\ref{theorem:arbitrary} shows instead that if we restrict ourselves to issue-wise majority, we do not need any restriction on players' goals in order for at least one good Nash equilibrium to survive. All in all, in specific classes of aggregators, although pre-vote negotiations cannot sustain all good equilibria, they {\em can} sustain some, even when removing any logical constraint on players' goals.


\section{Discussion \& Related Work} \label{sec:related}

In this section we discuss our results from two points of view. First, we sketch how our results can be applied to the preservation of logical consistency when aggregation occurs on logically interconnected issues, which is the key problem of judgment aggregation. Second, we relate aggregation games to the influential notion of boolean game.

\subsection{Pre-vote Negotiations and Collective Consistency}\label{section:constraints}

We showcase an application of endogenous aggregation games to binary aggregation with constraints, or judgment aggregation \citep{endriss16judgment,GrossiPigozzi2014}, where individual ballots need to satisfy a logical formula, the {\em integrity constraint}, in to be considered feasible or admissible. In case each individual provides an admissible ballot, the obvious question is whether the outcome of a given aggregation rule will be admissible, as well. Here is an instance of this problem.

\begin{example}\label{example:discursive}
Consider the scenario in Table~\ref{figure:discursive_dilemma}. Suppose we impose the integrity constraint $p \rightarrow (q \vee r)$, making ballot $(1,0,0)$ inadmissible. 
All individual ballots in the example satisfy this requirement but the majority ballot does not. 
\end{example}
Paradoxical situations as those in Example~\ref{example:discursive} can be viewed as undesirable outcomes of aggregation games.
%
%
Building on the example, assume that each agent to have the following goals:
$\gamma'_A =  p, 
\gamma'_B =   q, 
\gamma'_C =  \neg r$.
Let $\pi_A=\pi_B=\pi_C$ be constant payoff functions. Observe that parties' goals are all consistent with the integrity constraint $r \rightarrow (p \vee q)$, and that the admissible ballot $(1,1,0)$ satisfies all of them. Given these goals, the profile in Table~\ref{figure:discursive_dilemma} shows a truthful NE that, however, does not satisfy neither the goal of party~$B$ nor the integrity constraint $p \rightarrow (q \vee r)$. However, this equilibrium is {\em not} surviving because, intuitively, party $B$ could transfer enough utility to party $C$ for it to vote for $q$. 


%
But the key question is whether we can guarantee that inadmissible equilibria do not survive.\footnote{That voting paradoxes can be studied from an equilibrium refinement perspective is an old but rather underexplored idea \citep{gueth91majority}.} 
It is easy to see that if the integrity constraint is implied by some player's goal---intuitively, the player internalises consistency itself as a goal---then $\N$-truthful and $\N$-efficient equilibria will satisfy the integrity constraint and, by our results, they will be surviving in games with well-behaved aggregators (Theorem \ref{theorem:survival}). Vice versa, since only equilibria survive which are efficient for some winning coalition, with some extra assumptions on the aggregator (Theorem \ref{prop:surviving}), only collectively consistent outcomes are generated by the aggregation. 


\subsection{Boolean Games}

Boolean games \citep{harrenstein01boolean} are a logic-based representation of strategic interaction, where a set of individuals $A=\{1, 2, \ldots, m\}$ is assigned control of a set of propositional variables $P=\{p_1, p_2, \ldots, p_k\}$. Specifically, propositions are partitioned among the agents, i.e., each agent is assigned unique control over a subset of them, and each agent can decide to set the propositional variables he or she controls to true or false, in such a way that the final outcome of the boolean game is determined by the agents' truth value assignment on the variables each of them controls. 
Finally, each agent is equipped with a goal formula, i.e., a formula of propositional logic over the set of variables $P$. Typically, although agents have control over some propositional variables, they might not be able to realise their goal formula on their own.


Boolean games  can be seen as a very basic form of aggregation games, as in Definition~\ref{def:JA1}. 
That is, a boolean game $\mathcal{B}$, defined over $A$ and $PS$, can be seen as an aggregation game of the following form:
$\mathcal A^{\mathcal{B}} =  \tuple{\N^{\mathcal{B}},\I^{\mathcal{B}}, F^{\mathcal{B}}, \set{\gamma^{\mathcal{B}}_i}_{i\in\N}, \set{\pi^{\mathcal{B}}_i}_{i\in\N}}$
where: $\N^{\mathcal{B}}=A$ is the set of players, $\I^{\mathcal{B}} = PS$ is the set of issues, $F^{\mathcal{B}}$ is a dictatorship on issue $j$ by $i$, for any issue $j\in \I$ controlled by $i\in N$, i.e., $F^{\mathcal{B}}(\B)(j)=  B_i(j)$, $\gamma^{\mathcal{B}}_i$ is the goal formula  for $i$ in $\mathcal{B}$, and each $\pi^{\mathcal{B}}_i$ is constant. So a boolean game can be seen as an aggregation game where each individual is a dictator for the variables he or she controls. The goals are formulated on the outcome of the individuals' assignments and the payoff function plays no role, i.e., it is constant.

By the above reduction we are able to import the following complexity bounds from the boolean games literature:
\begin{proposition}
Let $\A$ be an aggregation game with arbitrary goals $\gamma_i$ for each $i\in \N$ and ${\prof{B}}$ a ballot profile.
\begin{enumerate}
\item The problem of verifying whether, for some transfer profile $\tau$, $\prof{B}$ is a NE of $\tau(\A)$ is {\textit{co-NP hard}};
\item  The problem of verifying whether, for some transfer function $\tau$ and ballot profile $\prof{B}'$ that is a Nash equilibrium of $\tau(\A)$, we have that  $F(\prof{B}') \models \bigwedge_{i\in N}\gamma_i$ is {\textit{$\Sigma^{2}_p$ hard}};
\item  The problem of verifying whether, for some transfer function $\tau$ and all ballot profiles $\prof{B}'$ that are a Nash equilibrium of $\tau(\A)$, we have that  $F(\prof{B}') \models \bigwedge_{i\in N}\gamma_i$ is {\textit{$\Sigma^{2}_p$ hard}}.
\end{enumerate}
\end{proposition}
\begin{proof}
 The results follow from \cite[Proposition 1]{wooldridge13incentive}, \cite[Proposition 6]{wooldridge13incentive} and \cite[Proposition 14]{wooldridge13incentive} respectively, together with the translation given above. 
\end{proof}
Boolean games have been extended in a number of ways, some of which will be dealt with next. It is however worth mentioning the extension by \citep{Gerbrandy06} to cooperative structures where coalitions are able to share the control of a propositional variable. Although the purpose is to study the logical properties of shared control and not the property of social choice functions, \citep{Gerbrandy06}  basically works with aggregation games with arbitrary aggregators, but without goals and without payoff function.


\paragraph{Boolean games and incentive engineering}

A class of boolean games that is relevant to our framework is boolean games with arbitrary payoffs, i.e., aggregation games of the form $\mathcal A^{\mathcal{B}} =  \tuple{\N^{\mathcal{B}},\I^{\mathcal{B}}, F^{\mathcal{B}}, \set{\gamma^{\mathcal{B}}_i}_{i\in\N}, \set{\pi^{\mathcal{B}}_i}_{i\in\N}}$ where the payoff is not necessarily constant. These boolean games have been introduced to account for efforts (or costs) in performing actions \citep{grant:2011b,wooldridge13incentive,Turrini16,harrenstein_etal:2014a}. When comparing two outcomes, a player will prefer the ones satisfying the goal, but will otherwise look at minimising the effort. This amounts to the same idea of having a payoff that is taken into account only in case goal satisfaction cannot discriminate between outcomes.

Boolean games with costs have been looked at from the point of view of {\em incentive engineering}, allowing payoffs to be manipulable, either by exogenous taxation mechanisms as in the work of \citep{wooldridge13incentive} and  \citep{harrenstein_etal:2014a}, or by endogenous negotiation as in the work of \citep{Turrini16}. In the exogenous setting, an external system designer can impose taxes on players' actions, by effectively influencing their decision-making towards the realisation of his or her own goal formula. In the endogenous setting, individuals undergo a pre-play negotiation phase and try to improve upon their final payoff using side-payments.

This has lead to the discovery of the existence of {\em hard equilibria} \citep{harrenstein_etal:2014a}, i.e., pure Nash-equilibria that cannot be removed by external incentives. Notice that their presence is even more frequent in endogenous settings, due to the fact that side-payments are a weaker form of manipulation than external intervention, as observed by \citep{Turrini16}. Endogenous boolean games are essentially endogenous aggregation games applied to a more restricted secting. However the idea of hard equilibrium does carry over to aggregation games in general. Think for instance of a situation in which there is only one issue, $p$, and only one winning coalition, $\N$. If everyone wants $p$ to be true, then the profile in which everyone votes for $p$ is an equilibrium that is impossible to remove by manipulating payoffs.

\section{Conclusions}\label{sec:conclusions}

The paper has proposed a model of pre-vote negotiation for games of binary aggregation. Although a number of papers in the literature on voting games have focused on the problem of avoiding undesirable equilibria, no model studying strategic behaviour in a pre-vote negotiation phase has so far been proposed. We used the model to show how pre-vote negotiations can restore the efficiency of truthful voting in such games. 

More specifically we established the following main results.
First, in uniform aggregation games for independent and monotonic aggregators where voters' goals are cubes, if an equilibrium is truthful and efficient it will be selected by equilibrium behaviour in the pre-vote negotiation phase (Theorem \ref{theorem:survival}). The `only if' variant of this claim holds with respect to the efficiency alone of the equilibria (Theorem \ref{prop:surviving}).
Second, in constant aggregation games for independent and monotonic aggregators where voters' goals are monotonic and aligned, there always exist a truthful and efficient equilibrium that will be selected by equilibrium behaviour in the pre-vote negotiation phase (Theorem \ref{coro:survival}).
Third, in uniform aggregation games for issue-wise majority, with arbitrary individual goals, there always exist a truthful and efficient equilibrium hat will be selected by equilibrium behaviour in the pre-vote negotiation phase (Theorem \ref{theorem:arbitrary}).

Our work is a first step towards the development of a body of theoretical results on how strategic interaction preceding voting influences the outcomes of group decision-making.

\section*{Acknowledgments}

Davide Grossi acknowledges support for this research by EPSRC under grant EP/M015815/1.
Paolo Turrini acknowledges support from Imperial College London under the Imperial College Research Fellowship ``Designing negotiation spaces for collective decision-making" (DoC AI1048).


\vskip 0.2in
\bibliographystyle{apalike}
\bibliography{jagames}

\end{document}